\documentclass[10pt, journal,compsoc]{IEEEtran}

\ifCLASSOPTIONcompsoc
  \usepackage[nocompress]{cite}
\else
  \usepackage{cite}
\fi

\usepackage{amsmath,url}
\usepackage{amssymb}
\usepackage{amsthm}
\usepackage{cases}
\usepackage{graphicx}
\usepackage{epstopdf}
\usepackage{tikz,multicol}
\usepackage{verbatim}
\usepackage{bm}
\usepackage{subcaption}
\usepackage{etoolbox,pgfplots}
\usepackage{mathtools}
\usepackage{diagbox}

{}
  \newtheorem{thm}{Theorem}
  \newtheorem{theorem}{Theorem}
  
  \newtheorem{cor}[thm]{Corollary}
  
  \newtheorem{definition}[thm]{Definition}
  
  \newcommand\numberthis{\addtocounter{equation}{1}\tag{\theequation}}

\def\E{\mathbb{E}} 
\def\R{\mathbb{R}}
\def\ie{{\em i.e.}}

\def\L{\mathcal{L}} 
\def\i{\mathbf{1}} 
\def\l{\ell}

\def\N{\mathcal{N}}

\def\S{\mathcal{S}}

\def\I{\mathcal{I}}

\def\Re{\mathcal{R}}

\def\G{\mathcal{G}}

\def\x{\textbf{x}}

\def\v{\textbf{v}}

\def\vt{\tilde{v}}
\def\l{\lambda}

\begin{document}

\title{Efficient CSMA using Regional Free Energy Approximations}


\author{Peruru~Subrahmanya~Swamy,
        Venkata~Pavan~Kumar~Bellam,
        Radha~Krishna~Ganti,
        and~Krishna~Jagannathan 
\IEEEcompsocitemizethanks{\IEEEcompsocthanksitem P. S. Swamy, R. K. Ganti, and K. Jagannathan are with the Department
of Electrical Engineering, IIT Madras, Chennai,
India, 600036.\protect\\
E-mail: \{p.swamy, rganti, krishnaj\}@ee.iitm.ac.in
\IEEEcompsocthanksitem V. P. K. Bellam is with Department of Electrical and Computer Engineering, Virginia Tech, Blacksburg, VA 24061, USA.\protect\\
E-mail: pavanbv@vt.edu}
\thanks{A part of this work \cite{spcom_kikuchi} has been presented at International Conference on Signal Processing and Communications (SPCOM) 2016, held at Bengaluru, India.}
}

\IEEEtitleabstractindextext{%
\begin{abstract}
CSMA (Carrier Sense Multiple Access) algorithms based on Gibbs sampling can achieve throughput optimality if certain parameters called the fugacities are appropriately chosen. However, the problem of computing these fugacities is NP-hard. In this work, we derive estimates of the fugacities by using a framework called the \emph{regional free energy approximations}. In particular, we derive explicit expressions for approximate fugacities corresponding to any feasible service rate vector. We further prove that our approximate fugacities are exact for the class of \emph{chordal} graphs.  A distinguishing feature of our work is that the regional approximations that we propose are tailored to conflict graphs with small cycles, which is a typical characteristic of wireless networks. Numerical results indicate that the fugacities obtained by the proposed method are quite accurate and significantly outperform the existing Bethe approximation based techniques.

\end{abstract}

}

\maketitle

\IEEEdisplaynontitleabstractindextext

\IEEEpeerreviewmaketitle
  
\IEEEraisesectionheading{\section{Introduction}\label{sec:introduction}}
 \IEEEPARstart{L}{ink} scheduling algorithms based on CSMA have received renewed attention, since they were recently proven to be throughput optimal~\cite{libin,qcsma,sinr_mimo}, in addition to being amenable to distributed implementation. The central idea lies in sampling the feasible schedules from a product form distribution, using a reversible Markov chain called the Gibbs sampler~\cite[Chapter 7]{bremaud}. In order to support a feasible service rate vector, certain parameters of the Gibbs distribution called the fugacities are to be appropriately chosen. However, the problem of computing the fugacities for a desired service rate vector is NP-hard~\cite{libin}. An  iterative algorithm based on stochastic approximation is proposed in~\cite{libin}, which asymptotically converges to the exact fugacities. However, the convergence time of this algorithm scales exponentially with the size of the network~\cite{fast_mixing}. Hence, it is not amenable to practical implementation in large networks.

In this work, we derive estimates of the fugacities by using \emph{regional free energy approximations}~\cite{yedidia} for the underlying Gibbs distribution. In particular, we derive explicit expressions for approximate fugacities corresponding to any feasible service rate vector. Since, the proposed method does not involve any iterative procedures, it does not suffer from convergence issues. A distinguishing feature of our work is that the regional approximations that we propose are tailored to conflict graphs with small cycles, which is a typical characteristic of wireless networks. We also prove that our approximate fugacities are exact for the class of \emph{chordal} graphs.

In the regional approximation framework, the vertices of the conflict graph are divided into smaller subsets called regions, and the Gibbs distribution is approximated  using certain locally consistent distributions over these regions. Specifically, the Gibbs free energy, which is a function of the Gibbs distribution, is approximated using a function of these local distributions called the Regional Free Energy (RFE). 

In this paper, we propose a general framework to find the approximate fugacities for any given conflict graph, and an arbitrary collection of regions. Our framework involves solving a regional entropy maximization problem, which is then used to find the approximate fugacities. While the approach we propose is applicable to any arbitrary choice of regions, the regional entropy maximization problem may not be amenable to a distributed solution. Furthermore, the accuracy of the approximation depends strongly on the choice of regions. Therefore a major challenge in obtaining a practically useful algorithm based on the regional approximation framework lies in judiciously choosing regions to ensure that the resulting CSMA algorithm is both accurate and amenable to distributed implementation. 

Spatial network models like the random geometric graphs \cite{baccelli_book} that are typically used to model the wireless networks contain many short cycles. For such topologies, traditional approaches based on Belief propagation and the Bethe free energy approximation \cite{bethe_jshin} tend to lead to inaccurate approximations of the fugacities, since they are tailored to tree-like structures. This motivates us to choose regions that explicitly include small cycles. In particular, we consider two choices of regions namely: i) \emph{Clique based regions}, which improves the error due to 3-cycles, and ii) \emph{Clique plus 4-cycle based regions}, which improves the error due to both 3, 4-cycles. 

Under these choice of regions, we derive a simple closed form solution for the regional maximum entropy problem. In particular, we first prove that the regional maximum entropy problem decouples into local entropy maximization problems over each region. Then, we derive closed form solutions for these local entropy maximization problems. Thus, the approximate fugacities can be explicitly computed without solving any optimization problems. We further prove that these approximate fugacities are exact for the class of chordal graphs.

We evaluate the performance of our algorithms over spatial networks modelled using geometric random graphs. Numerical results indicate that the fugacities obtained by the proposed method are quite accurate, and significantly outperform the existing Bethe approximation based techniques~\cite{bethe_jshin}.  

\subsection*{Related work}
In the literature, some special cases of the regional approximation framework are used to estimate the fugacities. Bethe approximation~\cite{bethe_jshin, allerton_bethe}  is one such special case.  The Bethe approximation is accurate when the underlying topology is tree-like. However, for spatial networks, which inherently contain many short cycles, the Bethe approximation~\cite{yedidia} may not be the right choice. An iterative approximation algorithm based on Inverse Generalized Belief Propagation (I-GBP) is proposed in \cite{bp_csma}. While the I-GBP considerably improves the performance in the presence of short loops, it suffers from convergence issues. In particular, it is not always guaranteed to converge, and hence not very reliable \cite{bethe_jshin}. In \cite{chordal}, closed form expressions for the fugacities are computed for chordal graphs using a technique called the Junction tree based decomposition \cite{book_martin}. 


The rest of the paper is organised as follows. In Section \ref{sec_model}, we introduce the system model. In Section \ref{sec_rfe}, we review the concept of Regional free energy approximation. In Section \ref{sec_inverse}, we propose a unified framework to compute the approximate fugacities under arbitrary choice of regions.  In Section \ref{sec_choice}, we investigate the  accuracy and complexity of different choices of regions. In Section \ref{sec_clique}, we present our results for clique-based regions. In Section \ref{sec_four_cycles}, we derive the results corresponding to clique plus 4-cycle based regions. In Section \ref{sec_simulations}, we present our numerical results.

\section{System model and Problem description} \label{sec_model}
We consider a single-hop wireless network with $N$ links. We use the letter $\N$ to denote the set of all the links in the network. We represent the network using the widely used conflict graph interference model ~\cite{libin, fast_mixing}. A conflict graph is an undirected graph $G(V,E)$, in which each vertex corresponds to a wireless link (Transmitter - Receiver pair), and two vertices share an edge if simultaneous transmissions from the corresponding wireless links result in a collision. For a given link $i \in V$, the neighbour set $\N_i := \{j : (i,j) \in E\}$ denotes the set of conflicting links. 

We consider a slotted time model, and use $\x(t)=[x_i(t)]_{i=1}^N \subseteq \{0,1\}^N$ to denote the transmission status (or \emph{schedule}) of the links in the network. Specifically,  if a link $i$ is scheduled to transmit in a given time slot $t$, then the link is said to be active, and $x_i(t)$ is set to $1$. We assume that an active link can transfer unit data in a given time slot, if none of its conflicting links are active in that slot. We define service rate of a link as the long-term fraction of time that the link is active.

\emph{Rate region:}  A schedule $\x$ is said to be \emph{feasible} if no conflicting links are active simultaneously. Hence, the set of feasible schedules is given by $\mathcal{I}:= \{\x \in \{0,1\}^N \; : \; x_i + x_j \leq 1, \; \forall (i,j) \in E\}$. Then the feasible rate region $\Lambda$, which is the set of all the possible service rates over the links, is the convex hull of $\I$ given by $\Lambda:= \{ \sum_{\x \in \I } \alpha_{\x} \x : \sum_{\x \in \I} \alpha_{\x}=1, \alpha_{\x} \geq 0, \forall \x \in \I \}.$  Next, we describe a basic CSMA algorithm~\cite{fast_mixing}.

\emph{Basic CSMA:} 
\label{subsec_acsma}
In this algorithm, each link $i$ is associated with a real-valued parameter $v_i \in \R$ (referred to as fugacity) which defines how aggressively a link captures the channel. In each time slot, one randomly selected link $i$ is allowed to update its schedule $x_i(t)$ based on the information in the previous slot:
\begin{itemize}
\item If the channel is sensed busy, \ie,  $\exists j \in \N_i$ such that $x_j(t-1)=1$, then $x_i(t)=0$.
\item Else,  $x_i(t)=1$ with probability $\frac{\exp(v_i)}{1+\exp(v_i)}$.
\end{itemize}
Except for the selected link $i$, all the other links do not update their schedule, \ie, $x_j(t)=x_j(t-1), \forall j \neq i$.
It can be shown that the above algorithm induces a Markov chain on the state space of feasible schedules. Further, the stationary distribution is a product-form Gibbs distribution \cite{fast_mixing} given by
\begin{align}
p(\x)&= \frac{1}{Z(\v)} \exp\left(\sum_{i=1}^N x_i v_i\right), & \forall \x \in \I \subseteq \{0,1\}^N,  \label{eq_dist}
\end{align}
where $\v=[v_i]_{i=1}^N$, and $Z(\v)$ is the normalization constant. Then, due to the ergodicity of the Markov chain, the long-term service rate of a link $i$ denoted by $s_i$ is equal to the marginal probability that link $i$ is active, \ie, $p(x_i=1)$. Thus, the service rates and the fugacity vector $\v$ are related as follows:
\begin{align}
s_i&= p_i(1) = \sum\limits_{\x \in \I \; : \; x_i=1 } Z(\v)^{-1} \exp\left(\sum_{i=1}^N x_i v_i\right), \; \; \forall i \in \N, \label{eq_serv_fug}
\end{align}
where $p_i(1)$ denotes $p(x_i=1)$.

\emph{Problem Description:} The CSMA algorithm can support any service rate in the rate region if appropriate fugacities are used for the underlying Gibbs distribution \cite{libin}. We consider the scenario where the links know their target service rates, and address the problem of computing the corresponding fugacity vector. In principle, the fugacities can be obtained by solving the system of equations in \eqref{eq_serv_fug}.  Unfortunately, for large networks, solving these equations is highly intractable since it involves computing the normalization constant $Z(\v)$ which has exponentially many terms.  In this work, we propose simple, distributed algorithms to efficiently estimate the fugacities. Our solution is inspired by the well known Regional free approximation framework \cite{yedidia} which is reviewed in the next section. 

\section{Review of Regional approximation framework} \label{sec_rfe}
We introduce the notion of Regional free energy (RFE), which is useful in finding accurate estimates of the marginals of a distribution like $p(\x)$ in \eqref{eq_dist}. We require the following definitions to introduce RFE \cite{book_martin}, \cite{v_anant_kik}, \cite{yedidia}.

\emph{Regions and Counting numbers:} For a given conflict graph $G(V,E)$, let $\Re \subseteq 2^V$ denote some collection of subsets of the vertices $V$. These subsets are referred to as \emph{regions}. Further, assume that each region $r\in \Re$ is associated with an integer $c_r$ called the \emph{counting number} of that region. 
A valid set of regions $\Re$, and the corresponding counting numbers $\{c_r\}$ should satisfy the following two basic rules: a) Each vertex in $V$ should be covered in at least one of the regions in $\Re$, \ie, $\bigcup_{r \in \Re} r = V$, b) For every vertex, the counting numbers of all the regions containing it, should sum to $1$, \ie, 
\begin{align}
\sum_{\{r \in \Re | i \in r\}} c_r =1, \; \; \forall i \in V. \label{eq_counting_general} 
\end{align}

\emph{Regional schedule:} The regional schedule at a region $r \in \Re$, denoted by $\x_r \in \{0,1\}^{r}$, is defined as the set of variables corresponding to the transmission status of the vertices in that region, \ie, $\x_r:=\left\lbrace x_k \; | \; k \in r  \right\rbrace.$ Let $\I_r:= \{\x_r \; | \; x_i + x_j \leq 1, \forall  (i,j) \in E\}$ denote the set of feasible regional schedules. The schedules in $\I_r$ are said to be locally feasible at the region $r$.

\emph{Regional distribution and entropy:} If $b$ denotes the probability mass function of the random variable $\x = [x_i]_{i=1}^N$, then the regional distribution $b_r$ denotes the marginal distribution of $b$ corresponding to $\x_r \subset \x$. In the special case of a region being a singleton set, \ie, $r=\{i\}$ for some $i \in V$, then we denote the corresponding marginal distribution as $b_i$ instead of $b_{\{i\}}$.  The entropy $H_r(b_r):=  -\sum_{\x_r}b_r(\x_r) \ln b_r(\x_r)$ is called the regional entropy at the region $r$.

\emph{Local consistency:} Let $\{b_r\}_{r \in \Re}$ shortly denoted as $\{b_r\}$, be some set of regional distributions which may not necessarily correspond to valid marginals of any distribution $b$. If every two regions $r, q \in \Re$ such that $r \subset q$, satisfy $\sum_{\x_q \setminus x_{r}} b_q (\x_q) = b_r(\x_r)$, $\forall \x_r$, then the set of distributions $\{b_r\}$ are said to be locally consistent, and are referred to as pseudo marginals. 


Now we use these definitions to introduce RFE. Assume that a valid collection of regions $\Re$, and the corresponding counting numbers are given.\footnote{We assume that all the singleton sets are present in the collection of the regions $\Re$. In case, this assumption is not satisfied, one can simply add those missing singleton sets with a counting number $0$.} Then the RFE for the CSMA distribution, with a fugacity vector $\v$,  is defined as follows. 
\begin{definition} (Regional Free Energy)
Let $\v$ be the fugacity vector of CSMA. Then, given a random variable $\x = [x_i]_{i=1}^N$ on the space of feasible schedules $\I$, and its probability distribution $b$, the RFE denoted by $F_{\Re}(b; \v)$ is defined as
\begin{align}
F_{\Re}(b; \v)= F_{\Re}(\{b_r\} ; \v) = U_{\Re}(\{b_r\}; \v) - H_{\Re}(\{b_r\}). \label{eq_RFE}
\end{align}
Here the first term, called the average energy, is given by the following weighted expectation
\begin{align}
U_{\Re}(\{b_r\}; \v)= - \sum_{r \in {\Re}}c_r \E_{b_r}\Big[\sum_{j \in r} v_j x_j\Big]. \label{eq_averageenergy}
\end{align}
The second term $H_{\Re}(\{b_r\})$, known as the \emph{Regional entropy}, is an approximation to the actual entropy $H(\x)$, and is given by
\begin{align}
H_{\Re}(\{b_r\})= \sum_{r \in {\Re}}c_r H_r (b_r). \label{eq_k_ent}
\end{align}
 \end{definition}

The stationary points\footnote{Here, constrained stationary point refers to a stationary point of the Lagrangian function of the RFE that enforces the local consistency contraints.}  of the RFE with respect to the regional distributions $\{b_r\}$, constrained over the set of psuedo marginals, provide accurate estimates \cite{book_martin} of the marginal distributions of $p(\x)$ in \eqref{eq_dist}. In other words, if the set of psuedo marginals $\{b^*_r\}$ is a stationary point of the RFE $F_{\Re}(\{b_r\} ; \v)$, then $\{b^*_r\}$ correspond to the estimates of the marginal distribution of $p(\x)$ over the respective regions. In particular, $b_i^*(x_i=1)$ provides an estimate of the service rate $s_i$  corresponding to the fugacity vector $\v$.


\section{Fugacity estimation using Regional approximation} \label{sec_inverse}

As discussed in Section \ref{sec_rfe}, the Regional approximation framework is generally used to estimate the service rates for given fugacities.
However, in this work, we use the Regional approximation not to estimate the service rates, but to design an algorithm to estimate the fugacities for given service rate requirements. To that end, we introduce the notion of Region approximated fugacities (RAF).
\begin{definition} (Region approximated fugacities)
Consider the RFE $F_{\Re}(\{b_r\}; \v)$ corresponding to some fugacities $\{v_i\}$. If there is a stationary point $\{b_r^*\}$ of the RFE $F_{\Re}(\{b_r\}; \v)$, such that its corresponding singleton marginals are equal to the desired service rates $\{s_i\}$, \ie, $b_i^*(x_i=1)=s_i, \; \; \; i=1\dots N$, then the fugacities $\{v_i\}$ are said to be Region approximated fugacities for the given service rates $\{s_i\}$.
\end{definition}

We propose a two step approach to compute Region approximated fugacities.
\begin{enumerate}
\item Find a set of locally consistent regional distributions $\{b_r\}$ with the following two properties:
\begin{enumerate}
\item[(P1.)] The marginals corresponding to the singleton regions are equal to the the desired service rates, \ie, $b_i(x_i=1)=s_i, \; \; \; i=1\dots N.$
\item[(P2.)] There exists a set of fugacities $\{v_i\}_{i=1}^N$ for which, the set of regional distributions $\{b_r\}$ is a stationary point of the corresponding RFE $F_{\Re}(\{b_r\}; \v).$
\end{enumerate}
\item Explicitly find those fugacities $\{v_i\}_{i=1}^N$ that satisfy P2.
\end{enumerate}

In Theorem \ref{thm_step1}, we state that obtaining a set of regional distributions $\{b_r\}$ with the properties P1 and P2, boils down to constrained maximization of the regional entropy.
 \begin{theorem} \label{thm_step1}
 Let $\{b_r^*\}$ be some local optimizer of the optimization problem \eqref{eq_opt_step1}. Then $\{b_r^*\}$ satisfies the properties P1 and P2.
 \begin{align}
 &\arg \max_{\{b_r\}}  \sum_{r \in \Re} c_r H_r(b_r), \; \; \; \text{subject to} \label{eq_opt_step1}\\
b_r(\x_r)& \geq 0, \; \x_r \in \I_r; \; \sum_{\x_r \in \I_r} b_r(\x_r) =1, \;r \in \Re, \label{eq_normalization} \\
\sum_{\x_q \setminus \x_{r}} b_q (\x_q) &= b_r(\x_r), \; \; \x_r \in \I_r, \text{ and } r, q \in \Re \text{ s.t. } r \subset q, \label{eq_consistency} \\
 b_i(1)&=s_i, \; \; \; i=1\dots N. \label{eq_bi_si}
\end{align}  
\end{theorem}
\begin{proof}
The property (P1) trivially follows from \eqref{eq_bi_si}. We provide the proof for property (P2) in Section \ref{proof_thm_step1}.
\end{proof}

Next, in Theorem \ref{thm_step2}, we propose a formula to compute the Region approximated fugacities.
\begin{theorem} \label{thm_step2}
 Let $\{b_r^*\}$, be a local optimizer of \eqref{eq_opt_step1}.  Then the fugacities
\begin{align}
\exp(\tilde{v}_i)&= \prod_{\{r \in \Re \; | \; i \in r\}} \left(\frac{b_r^*(\x_r^i)}{b_r^*(\bf{0})}\right)^{c_r}, \; \; \forall i \in \N, \label{eq_lf_gf}
\end{align}
are Region approximated fugacities for the desired service rates $\{s_i\}$. Here $\x_r^i$ denotes the argument $\x_r$ with $x_i=1, x_j=0, \forall j \in r \setminus \{i\}$, and $\bf{0}$ denotes the argument with $x_j=0, \forall j \in r$.
\end{theorem}
\begin{proof}
Proof is provided in Section \ref{proof_thm_step2}.
\end{proof}

The results derived in Theorem 1 and 2 provide a unified framework for computing the Region approximated fugacities for arbitrary choice of regions. However, one needs to solve the possibly non-convex maximization problem \eqref{eq_opt_step1}, which is a non-trivial problem in general. Hence, in the next section, we first investigate the effect of different choices of regions on the accuracy and complexity of computing RAF. Then, we propose distributed algorithms for some ``useful" choices of regions.




\section{Choice of regions} \label{sec_choice}
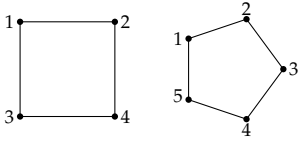
\begin{figure}
\begin{center}
\begin{tikzpicture}[scale=.7]
\draw
    (1,0)--(0.30,0.95)
    (0.30,0.95)--(-0.80,0.58) 
    (-0.80,0.58)--(-0.80,-0.58)
    (-0.80,-0.58)--(0.30,-0.95)
    (0.30,-0.95)--(1,0);
    
    \fill (1,0) circle(0.06cm);
    \fill (0.30,0.95)circle(0.06cm);
    \fill (-0.80,0.58)circle(0.06cm);
    \fill (-0.80,-0.58)circle(0.06cm);
    \fill (0.30,-0.95)circle(0.06cm);
    
    \draw (2-6,-.9) -- (2-6,.9) -- (3.8-6,.9) -- (3.8-6,-.9) -- (2-6,-.9); 
    
    \fill (2-6,-.9) circle(0.06cm);
    \fill (2-6,.9)circle(0.06cm);
    \fill (3.8-6,.9)circle(0.06cm);
    \fill (3.8-6,-.9)circle(0.06cm);
    
    
\node at (2-6-.2,-.9) {\begin{scriptsize}3\end{scriptsize}};
\node at (2-6-.2,.9) {\begin{scriptsize}1\end{scriptsize}};
\node at (3.8-6+.2,.9) {\begin{scriptsize}2\end{scriptsize}};
\node at (3.8-6+.2,-.9) {\begin{scriptsize}4\end{scriptsize}};
    

\node at (-0.80-0.2,0.58) {\begin{scriptsize}1\end{scriptsize}};
\node at (0.30,0.95+.2) {\begin{scriptsize}2\end{scriptsize}};
\node at (-0.80-.2,-0.5) {\begin{scriptsize}5\end{scriptsize}};
\node at (0.30,-0.95-.2) {\begin{scriptsize}4\end{scriptsize}};
\node at (1+.2,0) {\begin{scriptsize}3\end{scriptsize}};
\end{tikzpicture}
\end{center}
\caption{Illustration of Ring topology: 4-cycle, 5-cycle.}
\label{fig_ring}
\end{figure}
\begin{figure*}
\begin{minipage}{.5 \textwidth}
\centering
\includegraphics[scale=0.5]{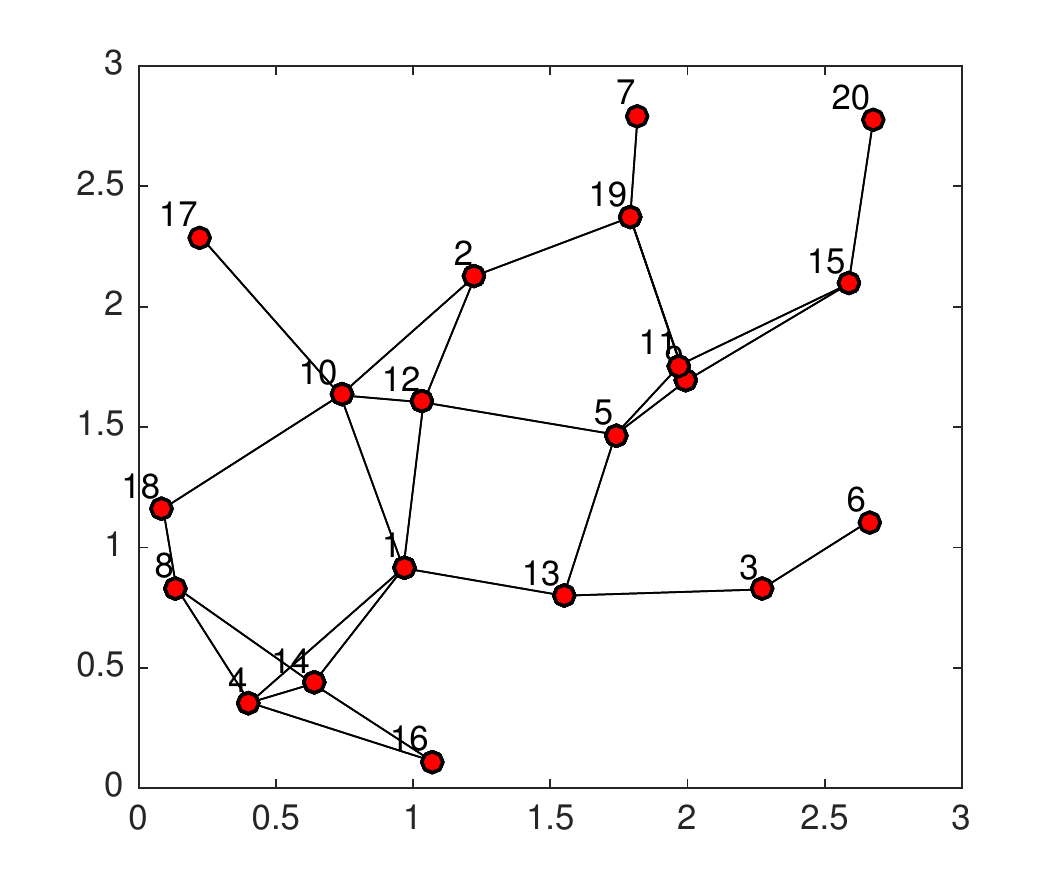}
\caption{A typical realization of a random geometric graph that is generally used to model wireless networks. Here, each vertex represents a wireless link, and two vertices are connected if they are within a certain distance called the interference radius.}
\label{fig_typical}
\end{minipage}
\begin{minipage}{.5 \textwidth}
\includegraphics[scale=0.5]{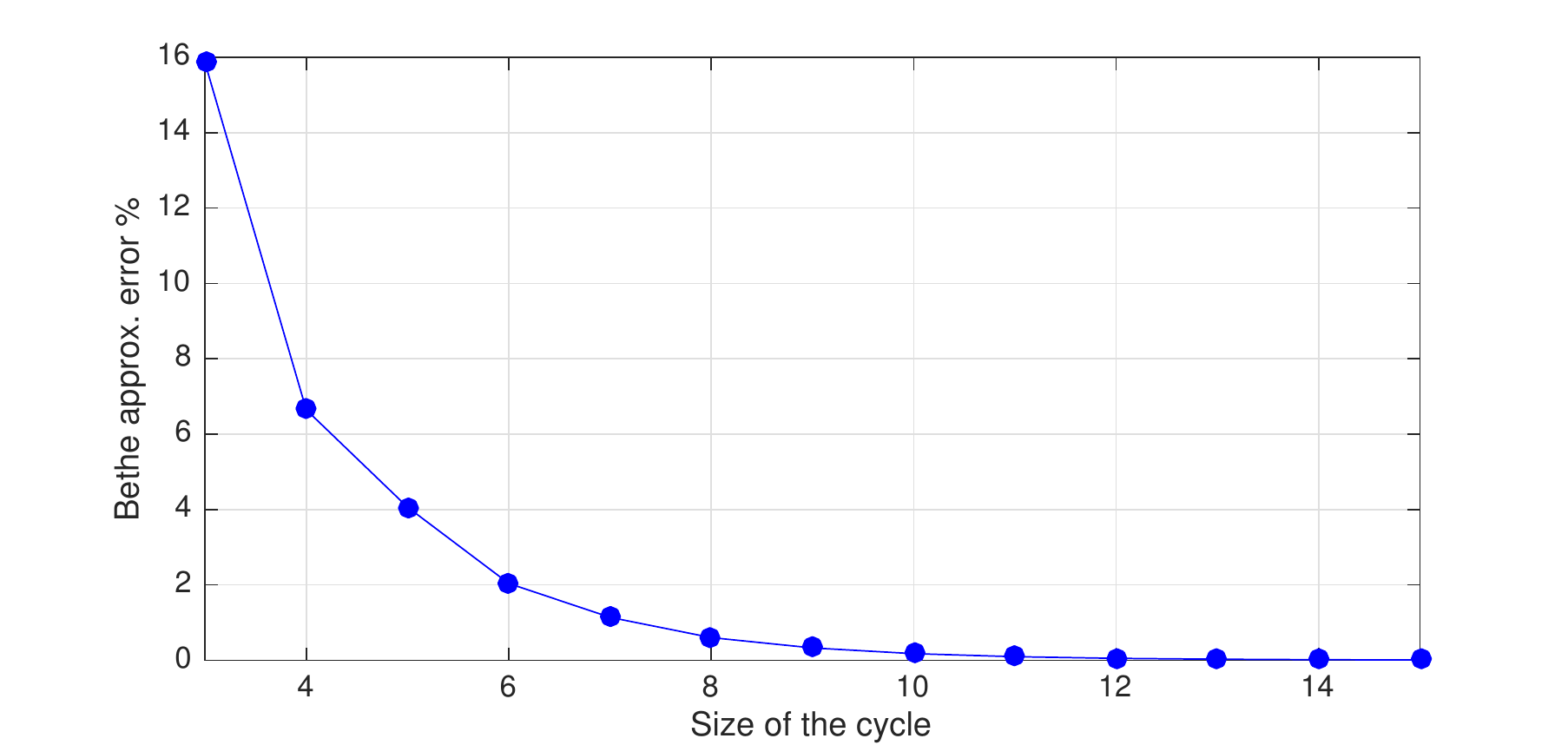}
\caption{Bethe approximation error for cycles of different sizes}
\label{fig_cycle}
\end{minipage}
\end{figure*}
The accuracy of the Regional approximation crucially depends on the choice of the regions. As the collection of regions becomes larger, the accuracy improves. However, for arbitrary choices of regions, the computation of RAF may not be amenable to a distributed implementation. Hence, the major challenge in using the Regional approximation framework for an algorithm like CSMA, is in choosing the regions that are as large as possible, while retaining the property of distributed implementation. 
For example, if we consider a special case of Regional approximation called the Bethe approximation \cite{yedidia}, simple distributed algorithms are known for the fugacities \cite{bethe_jshin} \cite{allerton_bethe}.

\subsection{Bethe approximation}
Under the Bethe approximation framework \cite{yedidia}, the collection of regions $\Re$ includes only the regions of cardinality one and two. Specifically, those are the regions corresponding to the vertices, and the edges of the conflict graph. Hence, the Bethe approximation is easy to implement. However, its accuracy is poor for topologies that contain many small cycles. Hence, Bethe approximation is not the right choice for typical wireless networks that inherently contain small cycles. (For example, Figure \ref{fig_typical} shows a typical realization of a random geometric graph that is generally used to a model wireless network. It can be observed that Figure \ref{fig_typical} contains small cycles.)

We consider a ring topology (see Figure \ref{fig_ring}), and plot the percentage error due to the Bethe approximation \cite[Section III]{bethe_jshin} as a function of the cycle size in Figure \ref{fig_cycle}. Here, the Bethe approximation error is defined as the loss in service rate that occurs due to the approximation used in computing the fugacitites (precise definition is provided later).  It can be observed from Figure \ref{fig_cycle}, that the significant error is mainly due to cycles of size $3,4$. For cycles of size $5$ or more, the error is within $5$ percent. Hence, if we can include the cycles of size $3,4$ in the collection of regions, the error can be significantly improved. To that end, we discuss a class of Regional approximations called the Kikuchi approximations, which can accommodate the cycles in the collection of regions.

\subsection{Kikuchi approximation}
Kikuchi approximation is a special case of Regional approximation framework. The Kikuchi approximation framework imposes certain restrictions on how the regions, and the counting numbers can be selected. Specifically, 
\begin{enumerate}
\item The collection of regions $\Re$ should be closed under intersection.
\item The counting numbers should satisfy
\begin{align}
\sum_{\{q \in \Re \;|\; r \subseteq q\}} c_q =1, \forall r \in \Re. \label{eq_counting0}
\end{align}
\end{enumerate}
\emph{Remark:} It can be observed that the constraints \eqref{eq_counting_general} imposed by the general Regional approximation are weaker than the constraints \eqref{eq_counting0} imposed by the Kikuchi approximation. In particular, Regional approximation imposes the constraints \eqref{eq_counting0} only on the singleton regions, \ie,
\begin{align}
\sum_{\{q \in \Re \;|\; \{i\} \subseteq q\}} c_q =1, \forall i \in V.
\end{align}
Cluster Variation Method (CVM) \cite{yedidia} is a well-known approach to obtain a collection of regions $\Re$, that satisfy the properties of the Kikuchi approximation framework. In this approach, we start with an initial set of regions $\Re_0$ called the maximal regions. Then $\Re$ is obtained by including the set of the maximal regions $\Re_0$, and all the regions that can be obtained by taking all the possible intersections of the maximal regions. Further, the counting numbers are computed as follows. First, all the maximal regions are assigned a counting number $1$. Then the counting numbers of the other regions are iteratively computed using
\begin{align}
c_r = 1 - \sum_{\{q \in \Re \;|\; r \subset q\}} c_q, \; \; \forall r \in \Re. \label{eq_counting_formula}
\end{align}

In the next section, we use this Kikuchi approximation framework to the improve the approximation error incurred due to 3-cycles. 

\section{Clique based regions to improve the error due to 3-cycles} \label{sec_clique}
Let $\Re_0$ be the set of all the maximal cliques in the conflict graph. Let $\Re$ denote the set of regions obtained by including the set of the maximal cliques $\Re_0$, and all the regions that can be obtained by taking all the possible intersections of maximal cliques. Note that this choice of regions ensures that the error-inducing 3-cycles (which are nothing but the cliques of size 3), are included in the collection of the regions $\Re$. Let the counting numbers be obtained by \eqref{eq_counting_formula}.

\subsection{Explicit formula for RAF}
In Theorem \ref{thm_clique}, we derive explicit formula for the RAF.

\begin{theorem} \label{thm_clique}
Let $\Re_0$ be the set of maximal cliques in the conflict graph. Let $\Re$ denote the collection of regions obtained by performing cluster variation method with $\Re_0$ as the initial set of regions. Let the counting numbers be defined by \eqref{eq_counting_formula}. Then, for any given service rate requirements $\{s_i\}$, the Region approximated fugacities are given by
\begin{align}
\exp(\tilde{v}_i)&= s_i  \prod_{\{r \in \Re \; | \; i \in r\}} \Big(1- \sum_{j \in r} s_j\Big)^{-c_r}, \; \; i \in \N. \label{eq_thm_clique}
\end{align}
\end{theorem}
\begin{proof}
Consider the optimization problem \eqref{eq_opt_step1} in Theorem \ref{thm_step1}. Now, we use the fact that any region $r \in \Re$ is a clique, and show that there is a unique set of $\{b_r(\x_r)\}$ that satisfy the feasibility constraints \eqref{eq_normalization} - \eqref{eq_bi_si} of the optimization problem \eqref{eq_opt_step1}. Let $\{b_r(\x_r)\}$ be feasible for \eqref{eq_opt_step1}. Since, every region $r \in \Re$ is a clique, from the definition of local feasibility, no more than one vertex can be simultaneously active in a feasible regional schedule $\x_r \in \I_r$. Hence, the set of local feasible schedules for a clique region $r$ is given by
\begin{align*}
\I_r&=\{\x_r^j\}_{j \in r} \cup \{\bf{0}\}.
\end{align*}
where $\x_r^j$ is the schedule with $x_j=1, x_k=0, \forall k\in r \setminus \{j\}$, and $\bf{0}$ is the schedule with $x_k=0, \forall k \in r$.
Then from  the local consistency constraint \eqref{eq_consistency} with $r=\{j\}$ and some $q \supset \{j\}$, we have
\begin{align*}
\sum_{\x_q \in \I_q : x_j=1} b_q(\x_q)&= b_{j}(1), \\
b_q(\x_q^j)&=b_{j}(1).
\end{align*}
Then from \eqref{eq_bi_si}, and the normalization constraint \eqref{eq_normalization}, we conclude that
\begin{align}
b_r(\x_r)=
\begin{cases}
      s_j,  &\text{if } \x_r=\x_r^j,  \\
      1- \sum_{k \in r} s_k,  &\text{if } \x_r=\bf{0}, \label{eq_reg_dist}
\end{cases}
\end{align}
is the only set of regional distributions $\{b_r(\x_r)\}$ that satisfy the feasibility constraints \eqref{eq_consistency} - \eqref{eq_bi_si}. Hence, it is the optimal solution of \eqref{eq_opt_step1}	. Substituting \eqref{eq_reg_dist} in Theorem \ref{thm_step2}, we obtain the RAF as follows:
\begin{align*}
\exp(\tilde{v}_i)&= \prod_{\{r \in \Re \; | \; i \in r\}} \left(\frac{s_i}{1-\sum_{j\in r}s_j}\right)^{c_r}, \; \; \forall i \in \N.
\end{align*}
Further, from \eqref{eq_counting_general}, we have $\sum_{\{r \in \Re | i \in r\}} c_r =1$. Hence, the RAF are given by \eqref{eq_thm_clique}.
\end{proof}

Next, we state a corollary of Theorem \ref{thm_clique} which gives the RAF corresponding to the Bethe approximation framework. This result is known due to \cite{bethe_jshin}, and we present it for the sake of completeness.
\begin{cor} \label{cor_bethe}
The Region approximated fugacities under the Bethe approximation framework are given by
\begin{align}
\exp(\tilde{v}_i)= \frac{s_i (1-s_i)^{|\N_i|-1}}{\prod_{j \in \N_i} (1-s_i-s_j)}, \; \; i \in \N, \label{eq_bethe}
\end{align}
where the set $\N_i$ denotes the neighbours of the vertex $i$ in the conflict graph.
\end{cor}
\begin{proof}
 As discussed earlier, in the Bethe approximation framework, the collection of regions $\Re$ includes only the regions corresponding to the edges and the vertices of the conflict graph, \ie, 
 \begin{align*}
  \Re= \lbrace\{i\}\; |\; i \in V \rbrace \cup \lbrace\{i,j\} \;| \; (i,j) \in E \rbrace.
 \end{align*}
 Since an edge is also a clique of the graph,  the observation \eqref{eq_reg_dist} is valid for the regions in the Bethe approximation framework. Hence, the result \eqref{eq_thm_clique} is applicable for this case.  Now, let us consider the regions of the Bethe approximation, and compute their counting numbers \eqref{eq_counting_formula}. Specifically, for a region $r=\{i, j\}$ corresponding to an edge $(i,j) \in E$ of the conflict graph, the counting number 
 \begin{align*}
c_{\{i,j\}}&= 1- \sum_{\{q \in \Re | \{i,j\} \subset q\}} c_q, \\
&= 1,
 \end{align*}
since a region corresponding to an edge $\{i,j\}$ is not a subset of any another edge or a vertex, \ie, $\{q \in \Re | \{i,j\} \subset q\}$ is an empty set.  

For a region $\{i\}$, corresponding to a vertex, the counting number 
\begin{align*}
c_{i}&= 1- \sum_{\{q \in \Re | \{i\} \subset q\}} c_q, \\
&= 1 - \sum_{j \in \N_i} c_{\{i,j\}},\\
&= 1- |\N_i|.
\end{align*}
Now, substituting these counting numbers in \eqref{eq_thm_clique} gives us the Bethe approximated fugacities \eqref{eq_bethe}.
\end{proof}

\subsection{Distributed Algorithm}
\begin{figure}
\begin{center}
  \hspace{-2mm}
\begin{tikzpicture}[scale=0.7]

\fill (0.5,0.5)circle(0.06cm);
\fill (3.5,0.5)circle(0.06cm);
\fill (1,2)circle(0.06cm);
\fill (3,2)circle(0.06cm);
\fill (4,2)circle(0.06cm);
\fill (1,1)circle(0.06cm);
\fill (3,1)circle(0.06cm);
\fill (4,1)circle(0.06cm);

\node at (.9,2.2) {\begin{scriptsize}8\end{scriptsize}};
\node at (1.1,.8) {\begin{scriptsize}2\end{scriptsize}};
\node at (3,2.2) {\begin{scriptsize}7\end{scriptsize}};
\node at (2.9,0.8) {\begin{scriptsize}3\end{scriptsize}};
\node at (4.1,2.2) {\begin{scriptsize}6\end{scriptsize}};
\node at (4.1,0.8) {\begin{scriptsize}5\end{scriptsize}};

\draw (1,2)-- (3,2) -- (3,1) -- (1,1) -- (1,2);  

\draw (3,2)-- (4,2) -- (4,1) -- (3,1) -- (3,2); 

\draw (1,1) -- (0.5,0.5); 
\node at (0.35,0.4) {\begin{scriptsize}1\end{scriptsize}};

\draw (3,1) -- (3.5,0.5); 
\node at (3.65,0.4) {\begin{scriptsize}4\end{scriptsize}};
;

\draw (1,1) -- (3,2); 

\draw (4,1) -- (3,2); 

\draw (3,1) -- (4,2); 
\end{tikzpicture}
  \vspace{-2mm}
 \caption{An example of a conflict graph with 8 vertices} 
  \label{fig_cliq}
  \end{center}
  \end{figure}
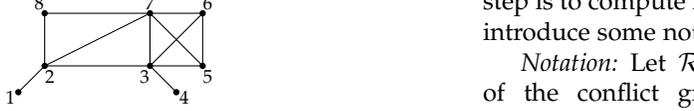
  
  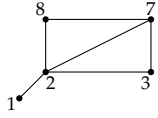
\begin{figure}
\begin{center}
  \hspace{-2mm}
\begin{tikzpicture}[scale=0.7]

\fill (0.5,0.5)circle(0.06cm);
\fill (1,2)circle(0.06cm);
\fill (3,2)circle(0.06cm);
\fill (1,1)circle(0.06cm);
\fill (3,1)circle(0.06cm);

\node at (.9,2.2) {\begin{scriptsize}8\end{scriptsize}};
\node at (1.1,.8) {\begin{scriptsize}2\end{scriptsize}};
\node at (3,2.2) {\begin{scriptsize}7\end{scriptsize}};
\node at (2.9,0.8) {\begin{scriptsize}3\end{scriptsize}};

\draw (1,2)-- (3,2) -- (3,1) -- (1,1) -- (1,2);  


\draw (1,1) -- (0.5,0.5); 
\node at (0.35,0.4) {\begin{scriptsize}1\end{scriptsize}};

;

\draw (1,1) -- (3,2); 


\end{tikzpicture}
  \vspace{-2mm}
 \caption{Local neighbourhood topology at vertex $2$ of Figure \ref{fig_cliq}} 
  \label{fig_local_topology}
  \end{center}
  \end{figure}
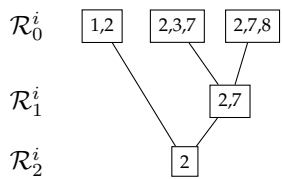
\begin{figure}
  \centering
  \begin{tikzpicture}[scale=1]

\node at (4,2) {$\Re_0^i$};
\node at (4,1) {$\Re_1^i$};
\node at (4,0.2) {$\Re_2^i$};

\node (rect12) at (5,2) [draw,minimum width=.25cm,minimum height=.25cm] {\begin{scriptsize}1,2\end{scriptsize}};
\node (rect237) at (6,2) [draw,minimum width=.25cm,minimum height=.25cm] {\begin{scriptsize}2,3,7\end{scriptsize}};
\node (rect278) at (7,2) [draw,minimum width=.25cm,minimum height=.25cm] {\begin{scriptsize}2,7,8\end{scriptsize}};

\node (rect27) at (6.7,1) [draw,minimum width=.25cm,minimum height=.25cm] {\begin{scriptsize}2,7\end{scriptsize}};


\node (rect2) at (6.1,0.2) [draw,minimum width=.25cm,minimum height=.25cm] {\begin{scriptsize}2\end{scriptsize}};



  \path (rect237) edge (rect27);
    \path (rect278) edge (rect27);
        \path (rect12) edge (rect2);
        \path (rect27) edge (rect2);

\end{tikzpicture}
\captionof{figure}{Pictorial representation of the intersections of the cliques at vertex $i=2$.}
  \label{fig_reg}
\end{figure}

We now propose a distributed algorithm to estimate the fugacities using cluster variation method with cliques as the maximal regions. Each link in the network can independently execute the algorithm once it obtains the a) target service rates of the neighbours, b) the local one-hop neighbourhood topology (See Figures \ref{fig_cliq}, \ref{fig_local_topology} for an example).

There are mainly two steps in the algorithm at a link $i$. The first step involves computing the maximal cliques, and their intersections in which the link $i$ is part of. The next step is to compute its fugacity by using the formula \eqref{eq_thm_clique}. We introduce some notations before we present the algorithm.

\emph{Notation:} Let $\Re_0^i$ be the collection of maximal cliques of the conflict graph $G(V,E)$ in which the vertex $i$ is part of. For example, if we consider the graph in Figure \ref{fig_cliq},  then $\Re_0^2 =\{\{1,2\}, \{2,8,7\}$, $\{2,3,7\}\}$. Similarly $\Re_0^3= \{\{2,3,7\}, \{3,5,6,7\}, \{3,4\} \}$. Using the information about the local topology, any standard algorithm like \cite{cliques_complexity} can be used for finding this set of maximal cliques $\Re_0^i$. For every maximal clique $r \in \Re_0^i$, let us associate a counting number $c_r=1.$
\noindent\rule[0.5ex]{\linewidth}{0.5pt}
\textbf{Algorithm 1:} Clique based distributed algorithm at link $i$\\
\noindent\rule[0.5ex]{\linewidth}{0.5pt}
\; \; \emph{Input:}  $\Re_0^i$, service rates $\{s_j\}_{j \in \N_i}$; \hspace{0.25cm}
\emph{Output:}  fugacity $\tilde{v}_i$.
\begin{enumerate}
\item Consider a variable $l$ called level, and initialize $l=0$.

\item Obtain $\Re_{l+1}^i$ by intersecting the cliques in level $l$, with the cliques in levels less than or equal to $l$, \ie,
\begin{align*}
\Re_{l+1}^i:= \{q_1 \cap q_2 \;|\;  q_1 \in \Re_{l}^i, q_2 \in \cup_{k \leq l} \Re_{k}^i, \; q_1 \neq q_2\}.
\end{align*}
 If  there are no intersections, \ie, $\Re_{l+1} = \Phi$, go to Step 6; Else continue.

\item From the set $\Re_{l+1}^i$, discard the following:
\begin{itemize}
\item[(i)] Cliques which are already present in a previous level, \ie, discard $r \in \Re_{l+1}^i$ if $r \in \cup_{k \leq l} \Re_{k}^i$.

\item[(ii)] Cliques which are  proper subsets of some other cliques in $\Re_{l+1}^{i}$, \ie, discard $r \in \Re_{l+1}^i$ if there exist any other set $q \in \Re_{l+1}^i$ such that $r \subset q$.
\end{itemize}

\item For each clique $r \in \Re_{l+1}^i$, compute 
\begin{align}
c_r= 1 - \sum_{q \in \S(r)} c_q, \label{eq_counting}
\end{align}
where $\S(r)=\{q \in \cup_{k \leq l}\Re^i_{k}\; | \; r \subset q\}$ is the set of cliques which are super sets of a given set $r$.
\item Increment $l$ by $1$, and go to step 2.

\item Let $\Re^i:= \cup_k \Re_k^i$ denote the collection of all the regions computed above. Then the fugacity is computed as
\begin{align}
\exp(\tilde{v}_i) = s_i \prod_{r \in {\Re^i}} \Big(1- \sum_{j \in r} s_j \Big)^{-c_r}.  \label{eq_main}
\end{align}
\end{enumerate}
\noindent\rule[0.5ex]{\linewidth}{0.5pt}

\emph{Complexity:} The worst case complexity incurred by a link to compute the corresponding maximal cliques $\Re_0^i$ is $O(3^{d/3})$, where $d$ is the the maximum degree of the graph \cite{cliques_complexity}. Hence, for spatial networks, where the degree of the graph does not scale with the network size, our algorithm computes the fugacities with $O(1)$ complexity. 

 \emph{Information exchange:} The information exchange required for our algorithm is very limited, since the algorithm is fully distributed except for obtaining the local topology information and neighbours' service rates.

 \emph{Example:} Let us consider the conflict graph shown in Figure \ref{fig_cliq}, and compute the fugacity for the vertex $2$ using the proposed algorithm. The  local topology, and the set of regions at vertex $2$ are shown in Figures \ref{fig_local_topology}, \ref{fig_reg}. The  set of maximal cliques containing the vertex $2$ is $\Re_0^2 =\{\{1,2\}, \{2,8,7\}$, $\{2,3,7\}\}$. Considering their intersections we get $\Re_1^2=\{\{2,7\}, \{2\}\}$. However, we discard the set $\{2\}$ since it is a proper subset of another region $\{2,7\}$ in the same level. Hence $\Re_1^2= \{\{2,7\}\}$. Now, $\Re_2^2$ is obtained by intersecting $\{2,7\}$ with regions in the previous level $\Re_0^2$. The only new region we obtain is the singleton set $\{2\}$, \ie,  $\Re_2^2=\{\{2\}\}$. 
 
Next, we compute the counting numbers of these regions. As defined earlier, all the maximal cliques $\Re_0^2$ will be given a counting number $1$. Further from \eqref{eq_counting}, it can be easily observed that for the set $r=\{2,7\}$, $c_r=-1$, since it has two super sets namely $q_1=\{2,8,7\}, \;q_2= \{2,3,7\}$  with $c_{q_1}=1, c_{q_2}=1$. Similarly, since every region in $\Re_0^2 \cup \Re_1^2$ is a super set of the region $r=\{2\}$, the counting number of $r=\{2\}$ is $c_r= 1- c_{\{2,7\}} - c_{\{1,2\}} - c_{\{2,8,7\}} - c_{\{2,3,7\}} = -1.$ Hence, the following expression gives the fugacity $\exp(\tilde{v}_2)$:
\begin{align*}
\frac{s_2  (1-s_2-s_7)(1-s_2)}{(1-s_1-s_2) (1-s_2-s_8-s_7) (1-s_2-s_3-s_7)}.
\end{align*}
Next, we prove that the Region approximated fugacities computed using the clique based approach are exact for a class of graphs called the chordal graphs.

\subsection{Exactness of Clique based approach}
It is known that the Bethe approximation is exact for tree graphs. In the Kikuchi approximation framework, we have considered larger collection of regions by including the maximal cliques of the graph. Hence, one may expect the accuracy to improve. Indeed, we confirm this intuition by proving that the clique based Kikuchi approximation is exact for a wider class of graphs called the chordal graphs (Trees are a special case of chordal graphs).

\begin{definition} \emph{(Chordal graph)}
A graph is said to be chordal if all cycles of four or more vertices have a chord. Here, a chord refers to an edge that is not part of the cycle but connects two vertices of the cycle. (See Figure \ref{fig_chordal}, in Section \ref{sec_simulations} for an example of a chordal graph.)
\end{definition}

\begin{theorem} \label{thm_chordal}
If the conflict graph is chordal, the formula proposed in \eqref{eq_main} gives the exact fugacities that correspond to the desired service rates, i.e., if we marginalize the $p(\x)$ \eqref{eq_dist} corresponding to the estimated fugacities $\tilde{v_i}$ \eqref{eq_main}, we obtain the required service rates.
\end{theorem}
\begin{proof}
Proof is provided in Section \ref{proof_chordal}.
\end{proof}

\subsubsection*{Complete graph topology}
Now, we consider the complete graph topology and compute the RAF. Further, since the complete graph is a chordal graph, the RAF are exact.
\begin{cor}
Let the conflict graph be a complete graph. Then, the fugacities that exactly support the desired service rates $\{s_i\}$ are given by
\begin{align*}
\exp(v_i)&= \frac{s_i}{1- \sum_{j \in \N} s_j}, \; \forall i \in \N.
\end{align*}
\end{cor}
\begin{proof}
In a complete graph, the only maximal clique constitutes the whole network. Hence $r = \N$ is the only region in the collection of regions $\Re$. Further, due to \eqref{eq_counting_formula}, the counting number $c_r=1$. Hence, from Theorem \ref{thm_clique}, \ref{thm_chordal} the result follows.
\end{proof}

\section{Improving the error due to 3-cycles and 4-cycles} \label{sec_four_cycles}
In Section \ref{sec_clique}, we have considered clique based regions to improve the error due to 3-cycles. In this section, we focus on improving the error due to 4-cycles. 
Specifically, in this 4-cycle based approach, our collection of regions at a vertex $i$ includes
\begin{itemize}
\item All the 4-cycle regions that include vertex $i$.
\item All the cliques that include vertex $i$.
\end{itemize}
Further, the counting numbers are computed using \eqref{eq_counting_formula}. 

\emph{Remark:} In this paper, we only deal with 4-cycles that do not have a chord. Whenever, we use the term 4-cycle, we mean a 4-cycle without a chord.

Observe that the above collection of regions does not fall under the standard CVM framework, since our collection of regions may not include all the intersections of the 4-cycle regions.  The advantage of proposing this collection of regions is that it results in a distributed solution for fugacities. In particular, as we will prove in Theorem \ref{thm_4cycle}, the regional entropy maximization problem \eqref{eq_opt_step1} which is possibly a non-convex problem, decouples into local entropy maximization problems at each region. Further, these local entropy maximization problems have a closed form solution as stated in Theorem \ref{thm_grid_formula}. We now formally state these results.

\begin{theorem} \label{thm_4cycle}
Let the set of all the 4-cycles, and all the cliques be denoted by $\Re_{4C}$, $\Re_{Cl}$ respectively. Let $\Re= \Re_{4C} \cup \Re_{Cl}$ be the collection of regions. Let $\{b_r^*\}$ be the solution of the regional entropy maximization problem \eqref{eq_opt_step1}. For the clique regions $r \in \Re_{Cl}$, $b_r^*$ is given by \eqref{eq_reg_dist}. For the 4-cycle regions $r \in \Re_{4C}$, $b_r^*$ is given by the solution of the following optimization problem:
\begin{align}
 &\underset{b_r}{\arg \max}  \; H_r(b_r), \; \; \text{subject to} \label{eq_opt_4cycles}\\
& b_r(\x_r) \geq 0, \; \x_r \in \I_r; \;\; \sum_{\x_r \in \I_r} b_r(\x_r) =1,\nonumber\\
&\sum_{\x_r \in \I_r : x_i=1} b_r (\x_r) = s_i, \; \; \; i \in r. \label{eq_const_4cycle}
 \end{align}
 \end{theorem}
 \begin{proof}
As obtained in the proof of Theorem \ref{thm_clique}, the regional distributions of the set of cliques $\{b_r^*\}_{r \in \Re_{Cl}}$ are explicitly determined by the service rates $\{s_i\}$ as shown in \eqref{eq_reg_dist}. Hence, in the objective of \eqref{eq_opt_step1}, the entropy terms corresponding to the cliques regions are completely determined by the service rates. In other words, the entropy terms corresponding to the 4-cycles are the only terms that are to be considered while maximizing the objective \eqref{eq_opt_step1}. Further, since no clique or another 4-cycle can be a super set of a 4-cycle, from \eqref{eq_counting_formula}, it can be observed that $c_r=1$ for all the 4-cycle regions $r \in \Re_{4C}$. Hence, the optimization problem \eqref{eq_opt_step1} can be effectively reduced to an optimization problem over $\{b_r\}_{4C}$ as follows:

\begin{align}
  \underset{\{b_r\}_{r \in \Re_{4C}}}{\arg \max}&  \sum_{r \in \Re_{4C}}  H_r(b_r), \; \; \; \text{s.t.} \label{eq_semi_local}\\
  b_r(\x_r) &\geq 0, \; \x_r \in \I_r; \;\; \sum_{\x_r \in \I_r} b_r(\x_r) =1, \; r \in \Re_{4C},\nonumber\\
\sum_{\x_q \setminus \x_{r}} b_q (x_q) &= b_r^*(\x_r), \; \; \x_r \in \I_r,  r \in \Re_{Cl}, q \in \Re_{4C}, \text{ s.t. } r \subset q. \label{eq_consistency_semi}
\end{align}

Now, let us look at the constraints of \eqref{eq_semi_local} involving $q \in \Re_{4C}$. It is easy to observe that for a given 4-cycle region (without a chord) $q \in \Re_{4C}$, the only possible subsets in the clique regions $r \in \Re_{Cl}$ are the regions corresponding to edges $\Re_{E}$, and singleton sets. 

Further, for any $r=\{i,j\} \in \Re_{E}$, the only feasible schedules are $\I_{\{i,j\}}=\{(0,0), (1,0), (0,1)\}$. Similarly, for a singleton region $r=\{i\}$, we have $\I_{i}=\{0,1\}$. From \eqref{eq_reg_dist}, we also know that $b^*_{\{i,j\}}(1,0)=s_i$, $b^*_{\{i,j\}}(0,1)=s_j$, and $b^*_i(1)=s_i$. Hence, the set of constraints in \eqref{eq_consistency_semi} for a given $q \in \Re_{4C}$ effectively reduces to
\begin{align}
\sum_{\x_q \in \I_q: x_i=1,x_j=0} b_q (\x_q) &= s_i, \; \; \{i,j\}\in \Re_{E}, \text{ s.t.} \{i,j\} \subset q,  \label{eq_mini1}\\
\sum_{\x_q \in \I_q: x_i=0,x_j=1} b_q (\x_q) &= s_j, \; \; \{i,j\}\in \Re_{E}, \text{ s.t.} \{i,j\} \subset q,  \label{eq_mini2}\\
\sum_{\x_q \in \I_q: x_i=1} b_q (\x_q) &= s_i, \; \; \forall i \in V, \text{ s.t. } \{i\} \subset q. \label{eq_mini3}
\end{align}
Note that we did not explicitly consider the constraints corresponding to $(0,0) \in \I_{\{i,j\}}$, $0 \in \I_{\{i\}}$, since they are implicitly captured by the normalization constraints in \eqref{eq_semi_local}.

For any feasible schedule of a 4-cycle region $\x_q \in \I_q$, by the definition of local feasibility, no more than one vertex across an edge can be active. Hence, for any $(i,j) \in E$ of the 4-cycle $q$, we have $\{\x_q \in \I_q | x_i=1, x_j=0\} = \{\x_q \in \I_q | x_i=1\}$. Hence, the constraints in \eqref{eq_mini1} - \eqref{eq_mini3} can be implicitly captured by \eqref{eq_mini3} alone.
Further, in \eqref{eq_semi_local}, there are no constraints that involve the optimization variables of different 4-cycle regions. Hence, the problem \eqref{eq_semi_local} decouples into independent optimization problem at each 4-cycle as given in \eqref{eq_opt_4cycles}.
\end{proof}
Now, we derive explicit formula for the solution of the optimization problem \eqref{eq_opt_4cycles}.
\begin{theorem} \label{thm_grid_formula}
Let $r \in \Re_{4C}$ be a 4-cycle region as shown in Figure \ref{fig_ring}. Then $b_r^*$, the solution of the optimization problem \eqref{eq_opt_4cycles} satisfies the expression \eqref{eq_loop_big} (displayed on the top of the next page).
\begin{figure*}[!t]
\scriptsize
\begin{align}
\label{eq_loop_big}
\begin{split}
\frac{b_r^*(\x_r^1)}{b_r^*({\bf{0}})}&=\frac{\sqrt{(s_1 (s_2+s_3-s_4-1)+s_2 (s_3+s_4-1)+(s_3-1) (s_4-1))^2+4 s_1 s_4 (s_1+s_4-1) (s_2+s_3-1)}}{2 (s_1+s_2-1) (s_1+s_3-1)}\\
&\; \;\;+\frac{-2 s_1^2-s_1 (s_2+s_3+s_4-3)-s_2 s_3-s_2 s_4+s_2-s_3 s_4+s_3+s_4-1}{2 (s_1+s_2-1) (s_1+s_3-1)}.
\end{split}
\end{align}
\hrulefill
\vspace*{4pt}
\end{figure*}
In particular, if we assume homogeneous service rates, \ie, if $s_i=s,\; \forall i \in r$, then we have
\begin{align*}
\frac{b_r^*(\x_r^i)}{b_r^*({\bf{0}})}&=\frac{-1 + 4 s + \sqrt{1 - 4 s + 8 s^2}}{2 - 4 s}, \; \; \forall i \in r,
\end{align*}
where $\x_r^i$ denote the argument with $x_i=1, x_j=0, \forall j \in r\setminus \{i\}$, and ${\bf{0}}$ denote the argument with $x_j=0, \forall j \in r$.
\end{theorem}
\begin{proof}
We prove for the case of homogeneous service rates. The proof can be easily extended to the general case. It can be shown \cite[Section 3.5]{libin_book} that the solution of a maximum entropy problem of the form \eqref{eq_opt_4cycles} is given by a product form distribution
\begin{align}
b_r^*(\x_r)& = \frac{\prod_{i \in r} \lambda^{x_i}}{Z}, \; \; \x_r \in \I_r, \label{eq_prod_4cycle}
\end{align}
for some $\lambda >0$. Here $Z$ is the normalization constant. The parameter $\lambda$ should be obtained by using the feasibility constraints \eqref{eq_const_4cycle}. Further observe that the ratio 
\begin{align}
\frac{b_r^*(\x_r^i)}{b_r^*({\bf{0}})} = \lambda, \;\;\forall i \in r. \label{eq_ratio}
\end{align}
For the 4-cycle region in Figure \ref{fig_ring}, the feasible schedules consists of  a) four schedules in which only one vertex is active, b) two schedules in which only the diagonal vertices are active. Using this structure in  \eqref{eq_const_4cycle}, \eqref{eq_prod_4cycle} we obtain 
\begin{align}
s= \frac{\lambda^2 + \lambda}{1+ 2 \lambda^2 + 4 \lambda}. \label{eq_quadratic}
\end{align}
Solving the quadratic equation \eqref{eq_quadratic}, and using the observation \eqref{eq_ratio} gives the required result.
\end{proof}

\emph{Remark:} In Theorem \ref{thm_grid_formula}, we obtained a closed form expression for the regional distributions corresponding to 4-cycles. We already got closed form expressions for the regional distributions corresponding to cliques in \eqref{eq_reg_dist}. Hence, the RAF \eqref{eq_lf_gf} under our 4-cycle based approach can be explicitly computed without solving any optimization problems.

Now to demonstrate our 4-cycle based approach, we consider the widely studied grid topology \cite{marbach}, and derive the RAF explicitly. A $4 \times 4$ grid topology is illustrated in Figure \ref{fig_grid}.
\begin{cor}
Let $s_i=s,\; \forall i \in \N$ be the service rate requirements for a grid topology. Then the RAF under the 4-cycle based method is given by
\begin{align*}
\exp(\tilde{v}_i)&=
\begin{cases}
\frac{\left(-1 + 4 s + \sqrt{1 - 4 s + 8 s^2}\right)^4}{16 (1-s) s^3}, \; \; \text{ if } |\N_i|=4, \\
\frac{\left(-1 + 4 s + \sqrt{1 - 4 s + 8 s^2}\right)^2}{4 s ( 1-2s)}, \; \; \text{ if } |\N_i|=3, \\
\frac{-1 + 4 s + \sqrt{1 - 4 s + 8 s^2}}{2 - 4 s},    \; \; \text{ if } |\N_i|=2.
\end{cases}
\end{align*}
\end{cor}
\begin{proof}
We prove the result for a vertex of degree $4$. Other cases can be obtained similarly. Consider a vertex $i$ with degree 4 (for example, vertex 11 in Figure \ref{fig_grid}). Then the vertex $11$ belongs to four 4-cycle regions namely $\{7,8,11,12\},\{11,12,15,16\},\{6,7,10,11\},\{10,11,14,15\}$ . Each of these 4-cycle regions will have a counting number \eqref{eq_counting_formula} equal to $1$. Further, the intersection of these 4-cycle regions results in 4 edge based regions namely $\{11,7\},\{11,10\},\{11,15\},\{11,12\}$. Since each edge region has two super sets in the form of 4-cycles, the counting number \eqref{eq_counting_formula} of each of these edges is equal to $-1$. Further, the intersection of the edges result in the singleton set containing the considered vertex $11$. Since, all the four edges, and the four 4-cycles are super sets for this singleton vertex, the counting number of this singleton set  computed using \eqref{eq_counting_formula}  is equal to $1$. Now, from Theorems \ref{thm_4cycle}, \ref{thm_grid_formula}, we know the expressions for the optimal regional distributions $\{b_r^*\}$ of all the regions. Substituting these expressions in the formula for the RAF \eqref{eq_lf_gf} derived in Theorem \ref{thm_step2} gives the required result.
\end{proof}

\emph{Remark on Implementation:}
A vertex in the conflict graph requires the two hop neighbourhood topology to determine the 4-cycle regions that the vertex is involved in. This two hop topology can be obtained by using a simple broadcast scheme with the neighbours. For example, see \cite[Section VI-C]{bp_csma} for details.




\section{Numerical Results} \label{sec_simulations}

\begin{figure*}
\begin{minipage}{.33 \textwidth}
\centering
\includegraphics[scale=0.34]{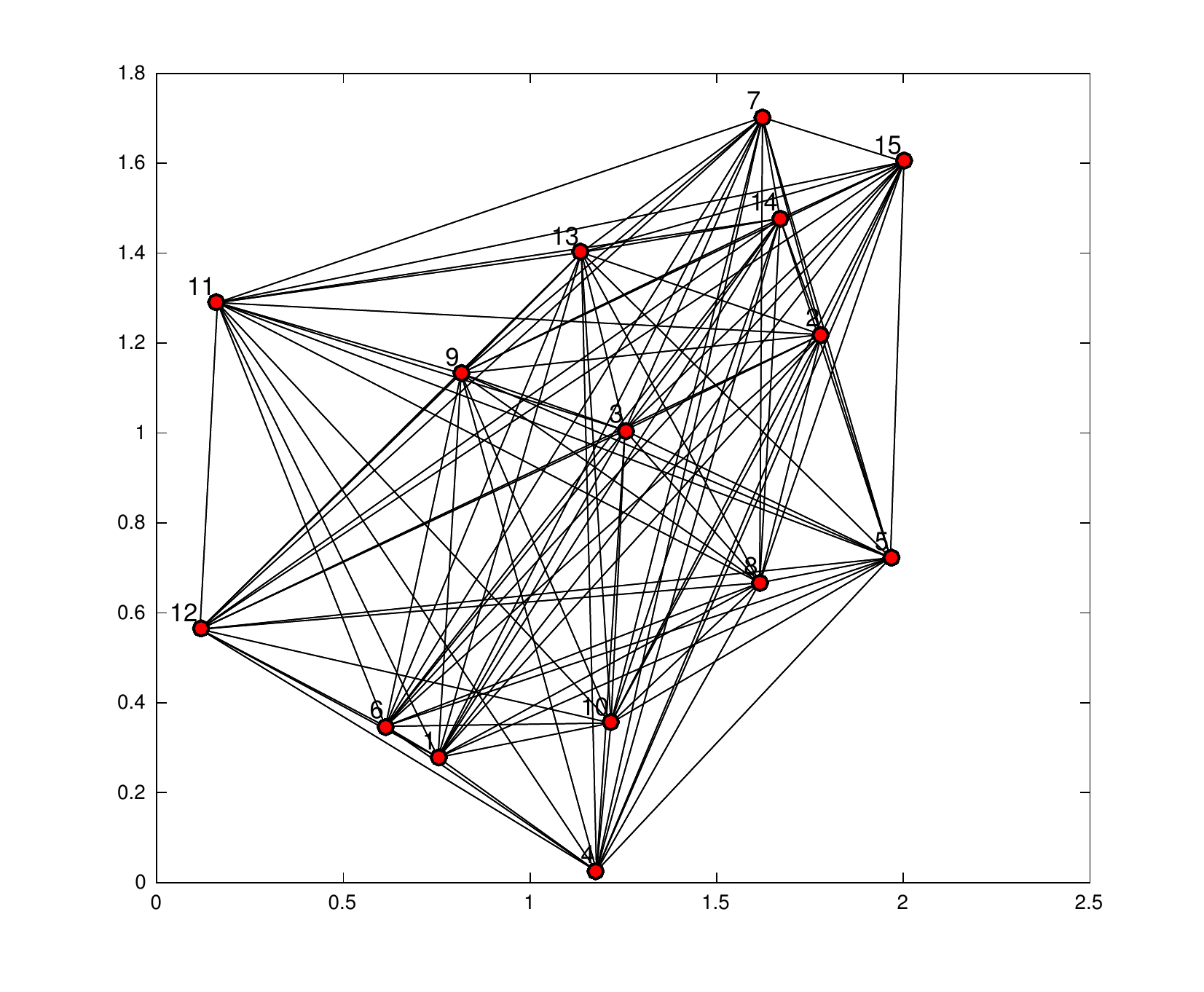} 
\caption{Complete graph}
\label{fig_complete}
\end{minipage}
\begin{minipage}{.33 \textwidth}
\centering
\includegraphics[scale=0.45]{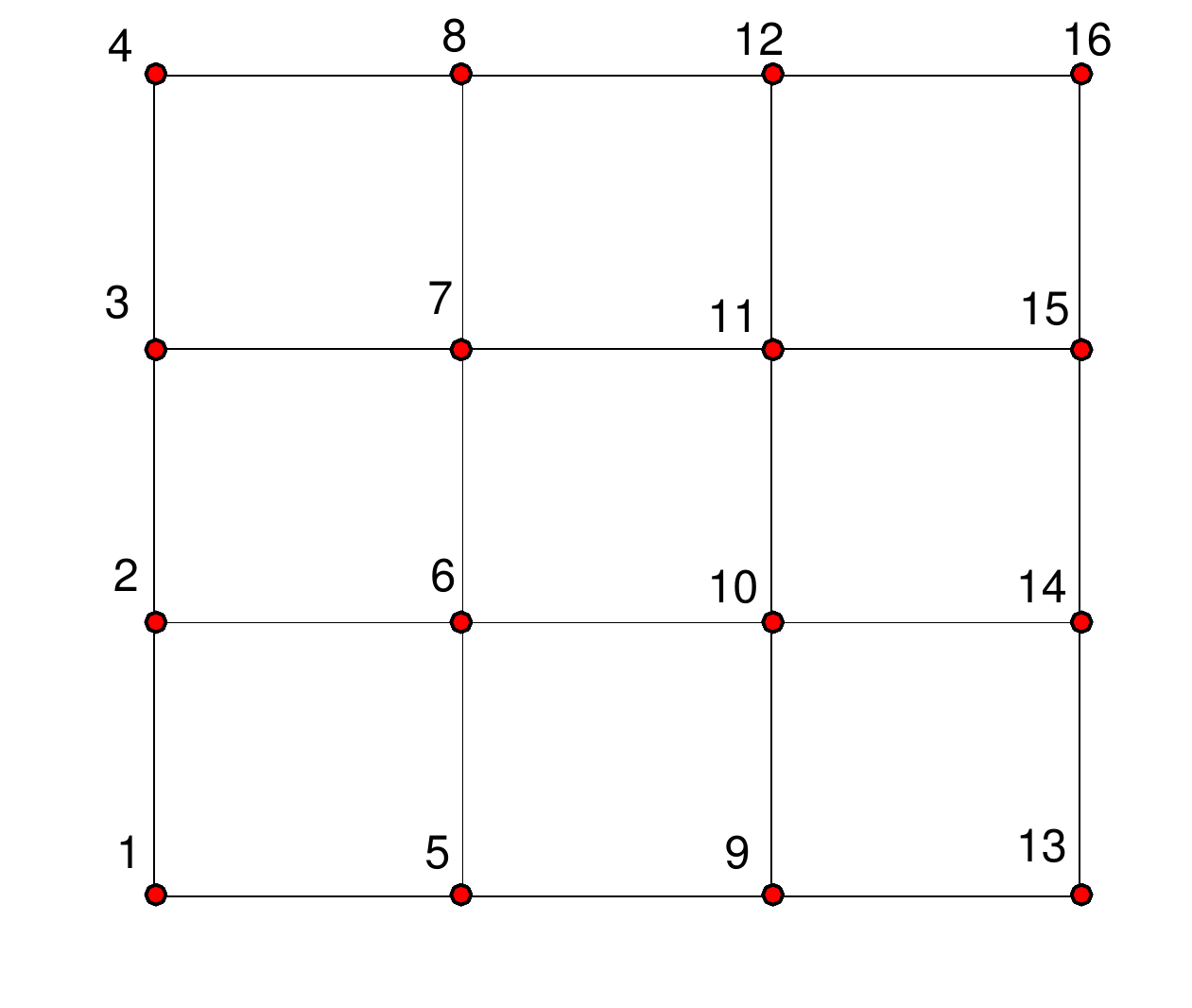} 
\caption{Grid graph}
\label{fig_grid}
\end{minipage}
\begin{minipage}{.33 \textwidth}
\centering
\includegraphics[scale=0.55]{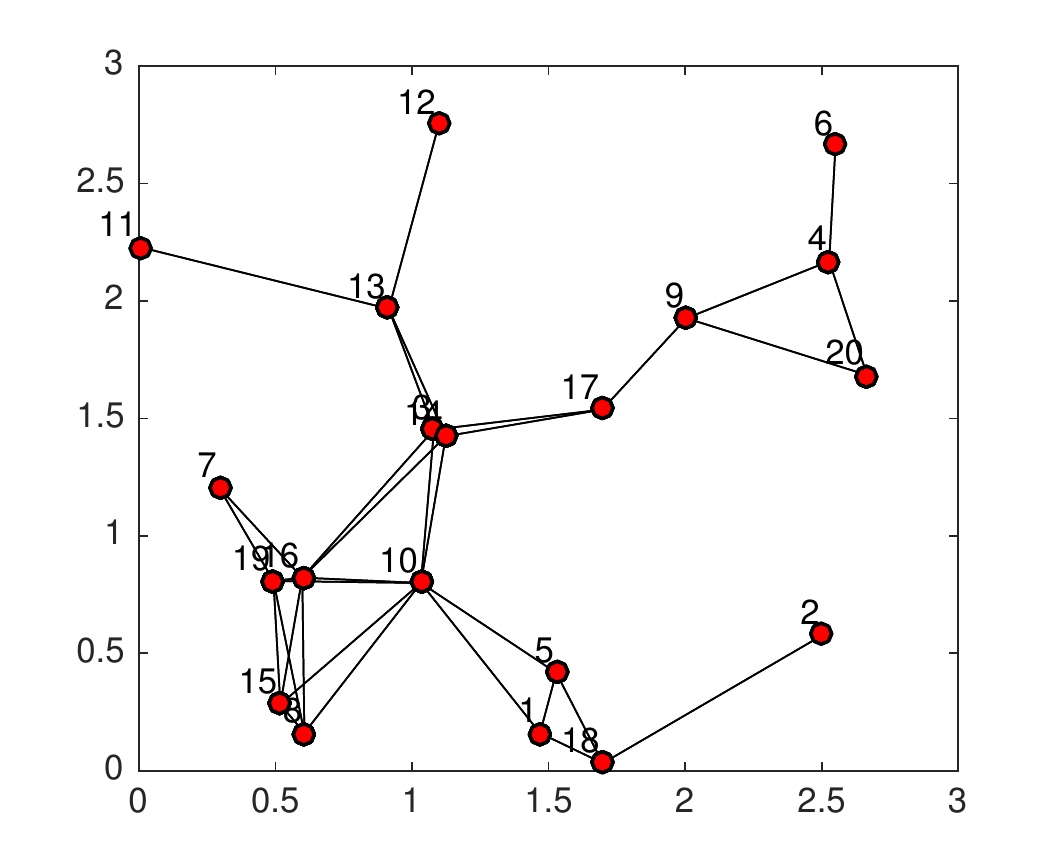} 
\caption{Chordal graph}
\label{fig_chordal}
\end{minipage}
\end{figure*}

\begin{figure*}
\begin{minipage}{.33 \textwidth}
\centering
\includegraphics[scale=0.42]{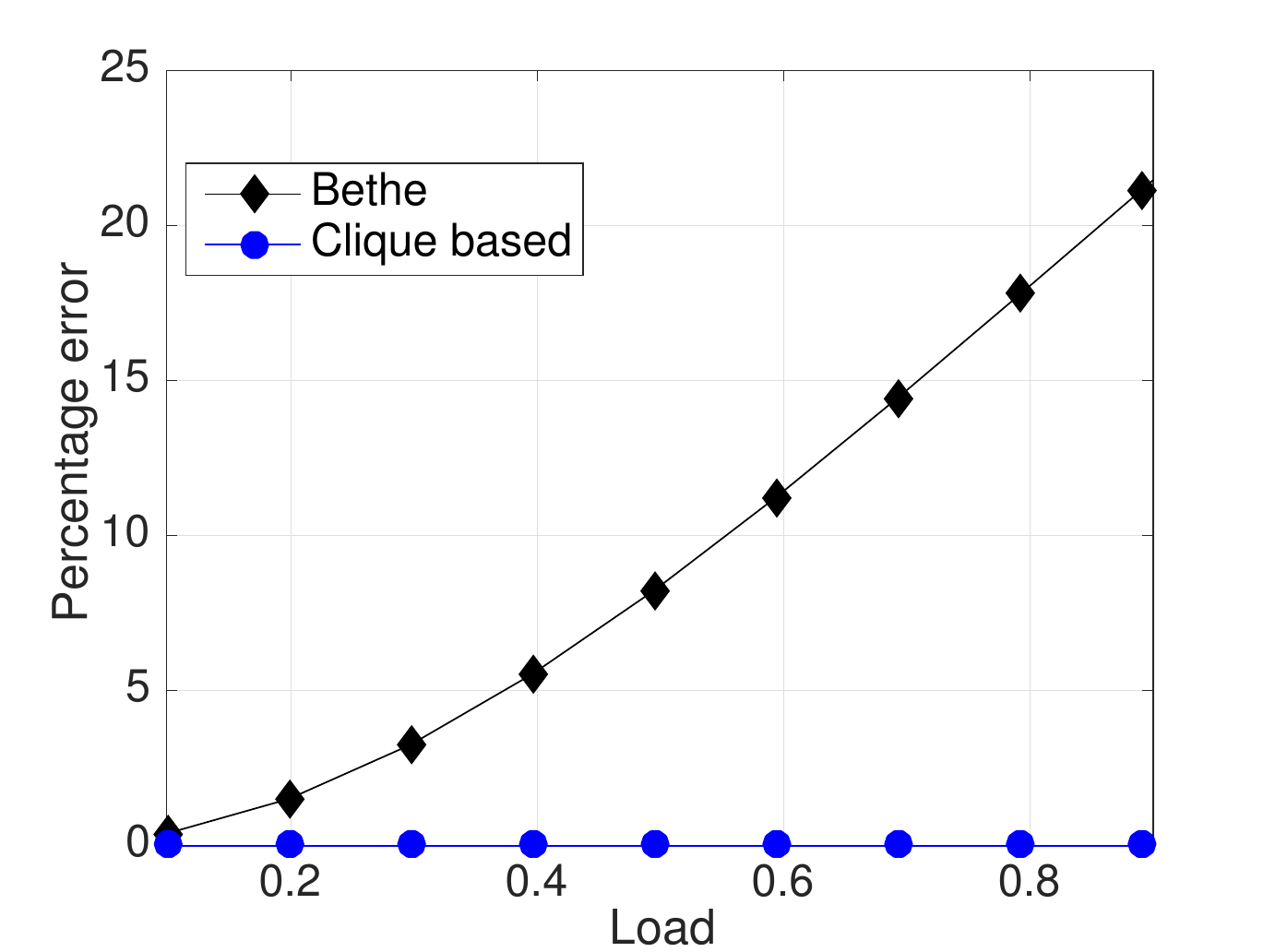} 
\caption{Error plot for Complete graph}
\label{fig_complete_prac}
\end{minipage}
\begin{minipage}{.33 \textwidth}
\centering
\includegraphics[scale=0.44]{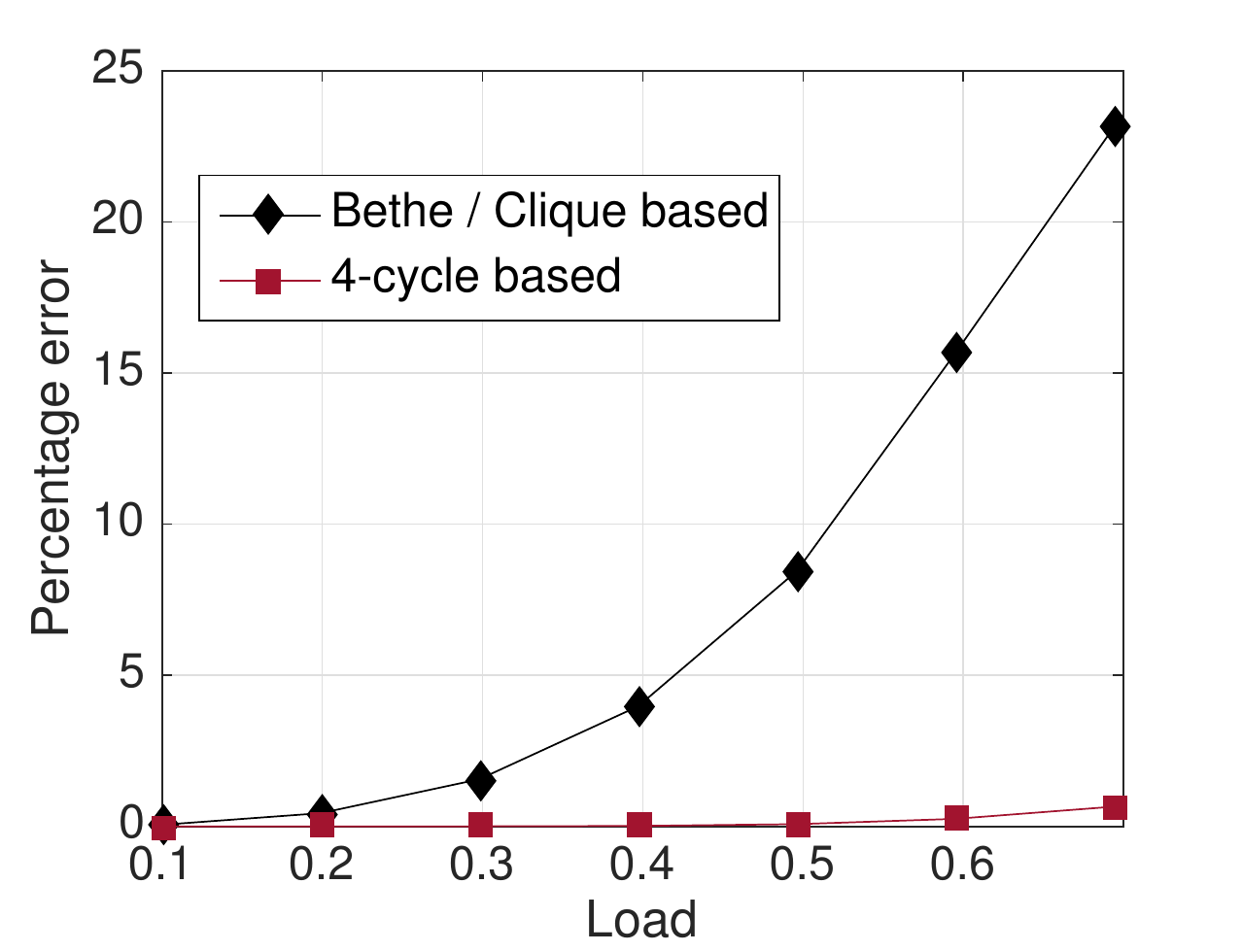} 
\caption{Error plot for grid graph}
\label{fig_grid_prac}
\end{minipage}
\begin{minipage}{.33 \textwidth}
\centering
\includegraphics[scale=0.42]{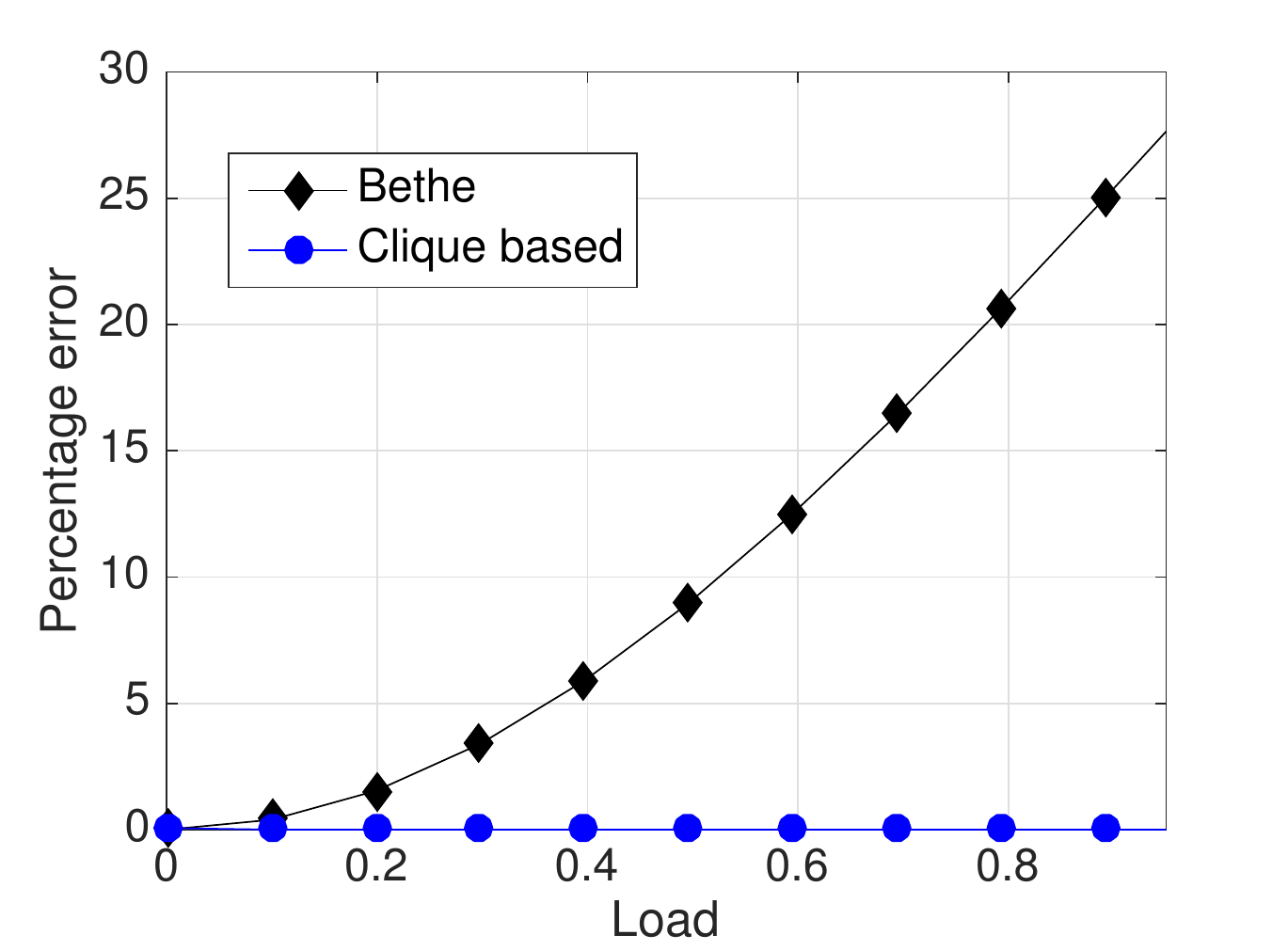} 
\caption{Error plot for chordal graph}
\label{fig_chordal_prac}
\end{minipage}
\caption*{The approximation error is plotted as function of the network load for the three topologies considered above}
\end{figure*}


\begin{figure*}
\begin{minipage}{.33 \textwidth}
\centering
\includegraphics[scale=0.39]{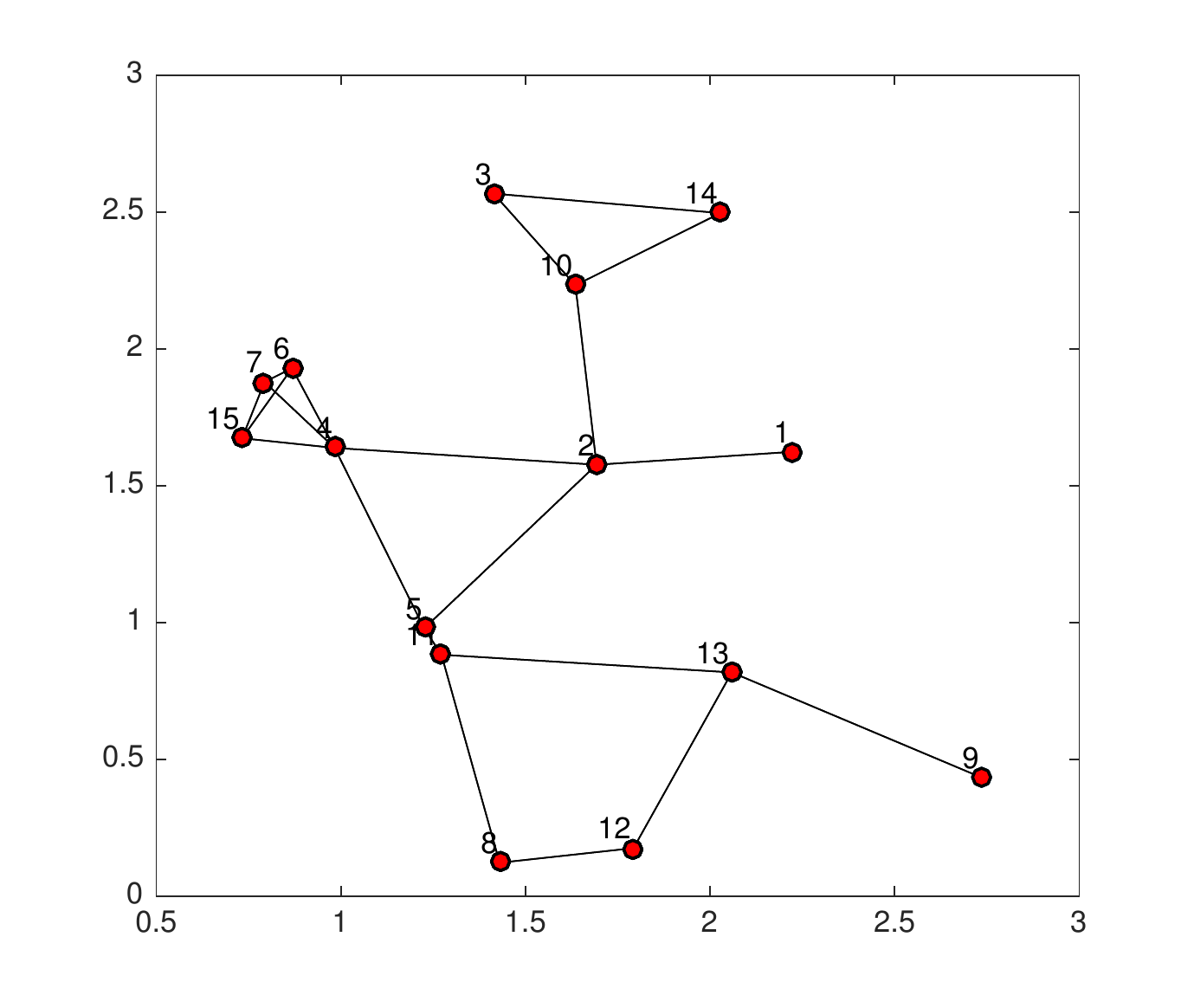} 
\caption{15-link random topology}
\label{fig_random15}
\end{minipage}
\begin{minipage}{.33 \textwidth}
\centering
\includegraphics[scale=0.465]{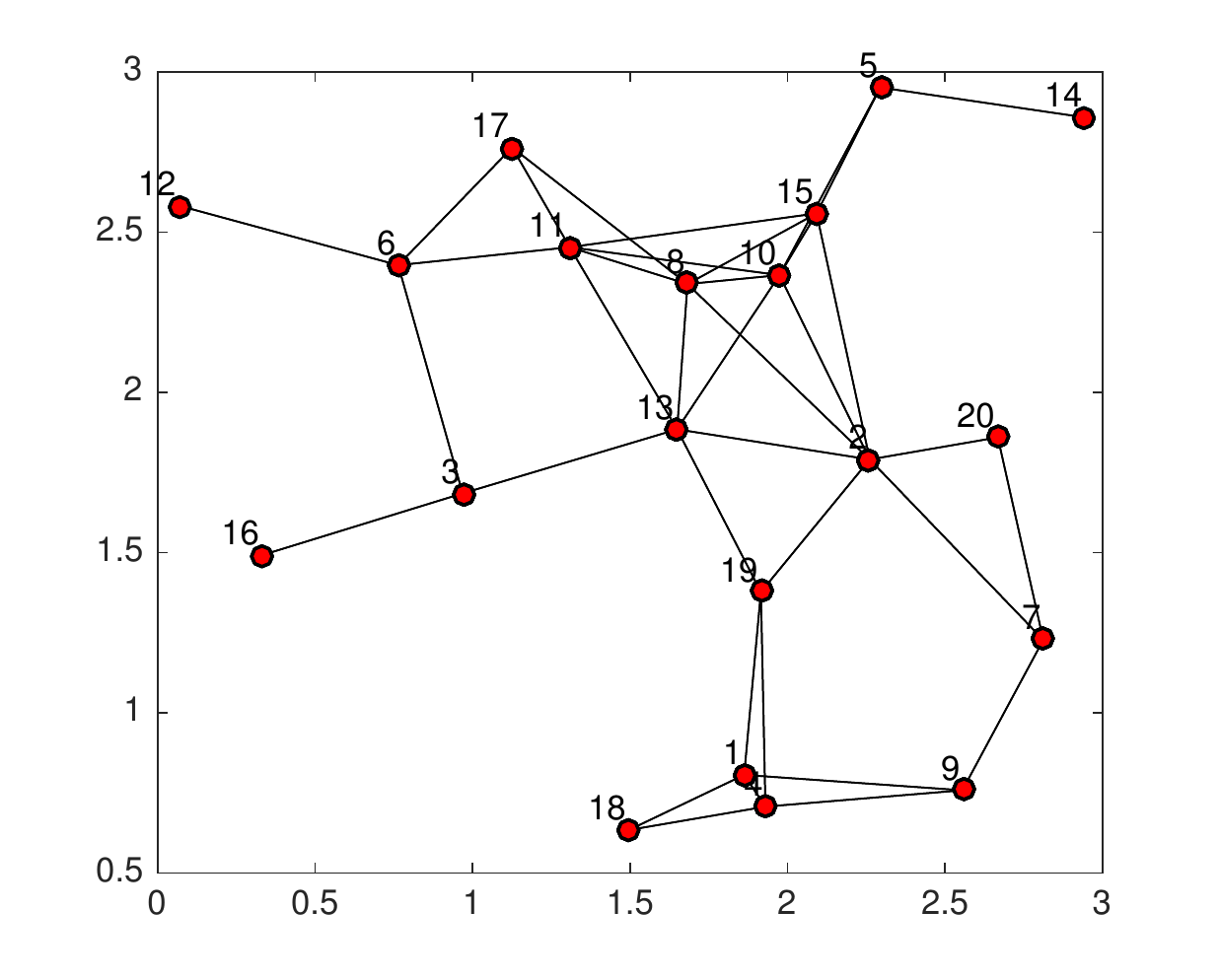} 
\caption{20-link random topology}
\label{fig_random20}
\end{minipage}
\begin{minipage}{.33 \textwidth}
\centering
\includegraphics[scale=0.43]{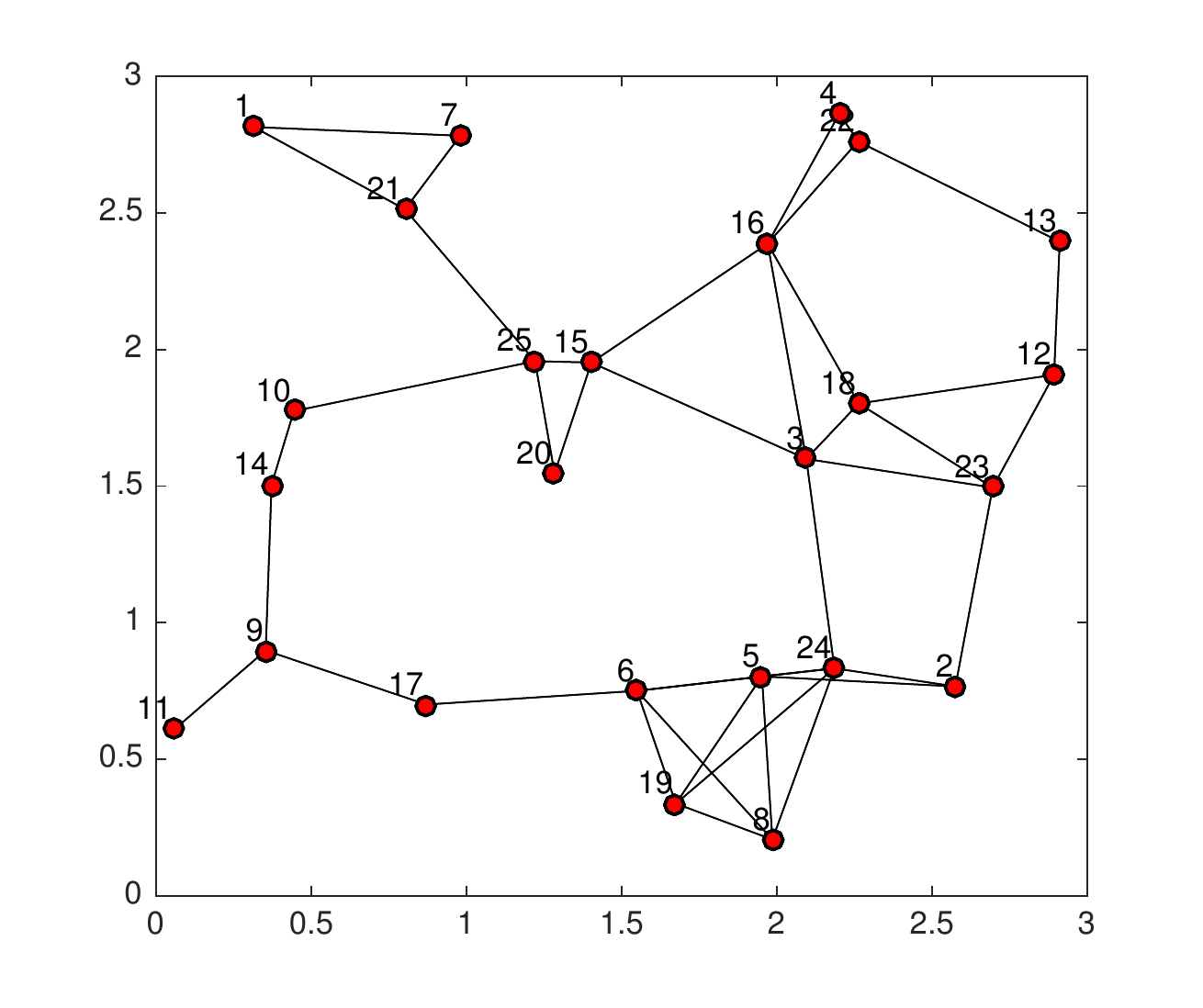} 
\caption{25-link random topology}
\label{fig_random25}
\end{minipage}
\end{figure*}

\begin{figure*}
\begin{minipage}{.33 \textwidth}
\centering
\includegraphics[scale=0.45]{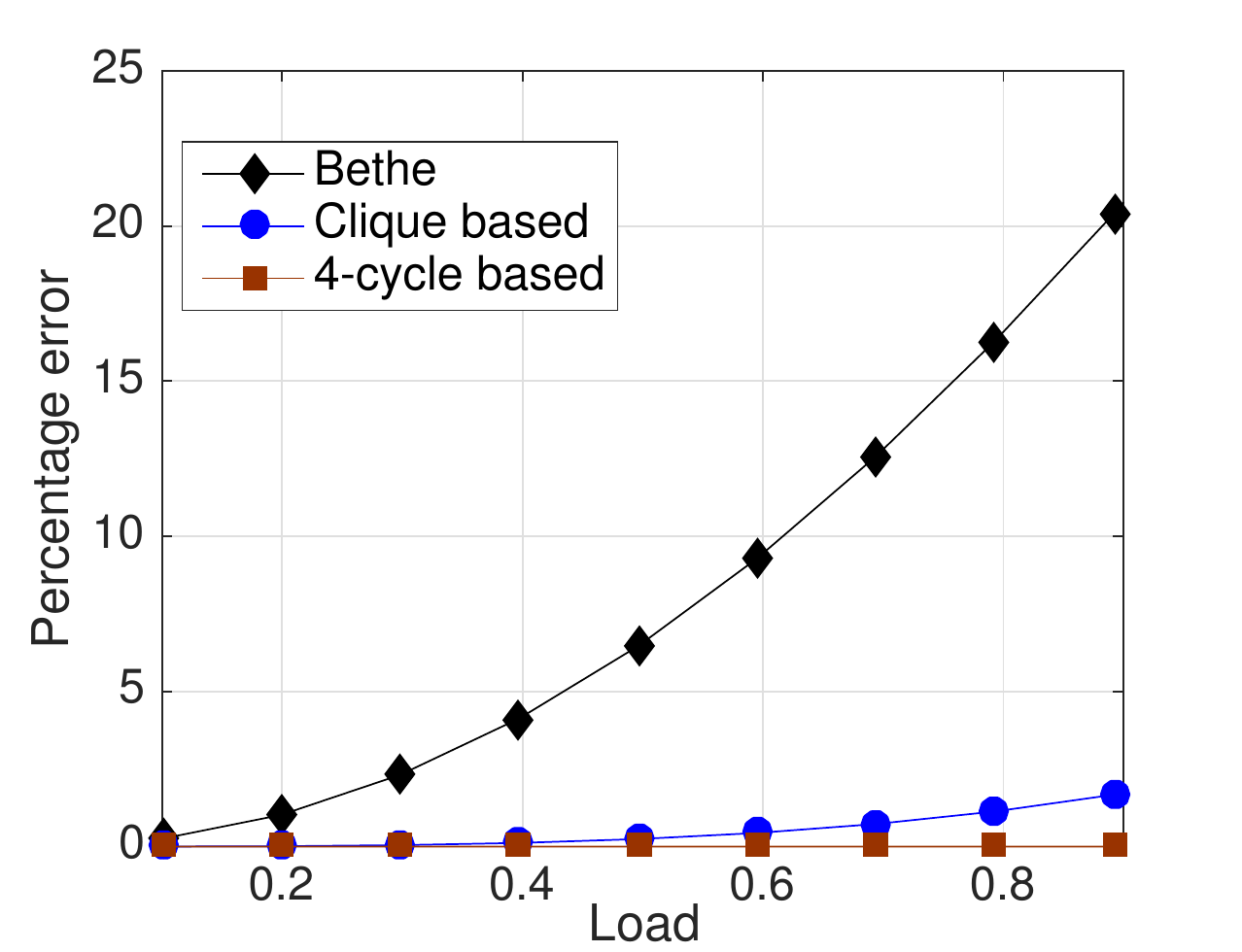} 
\caption{Error for 15-link topology}
\label{fig_r15_abs}
\end{minipage}
\begin{minipage}{.33 \textwidth}
\centering
\includegraphics[scale=0.45]{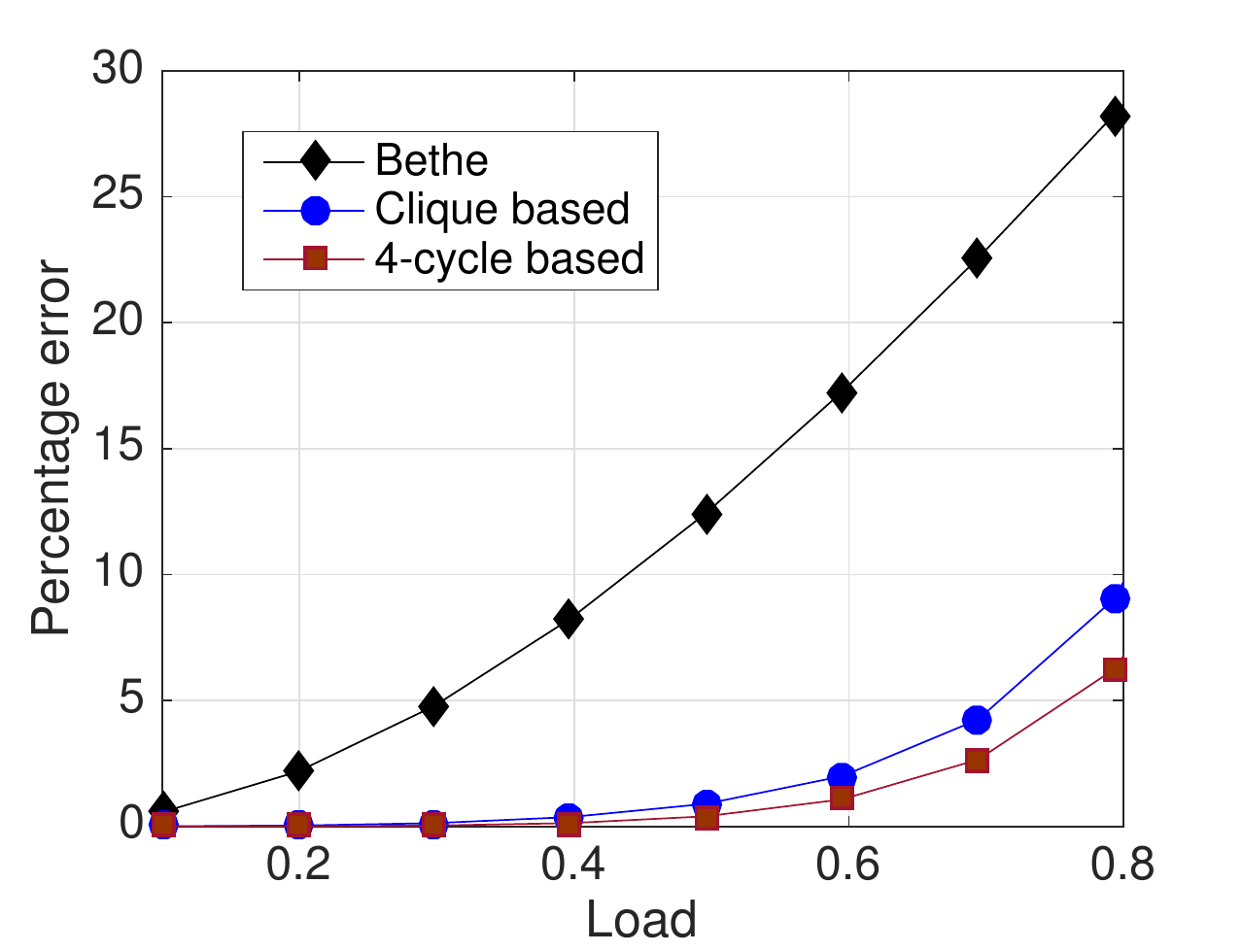} 
\caption{Error for 20-link topology}
\label{fig_r20_abs}
\end{minipage}
\begin{minipage}{.33 \textwidth}
\centering
\includegraphics[scale=0.45]{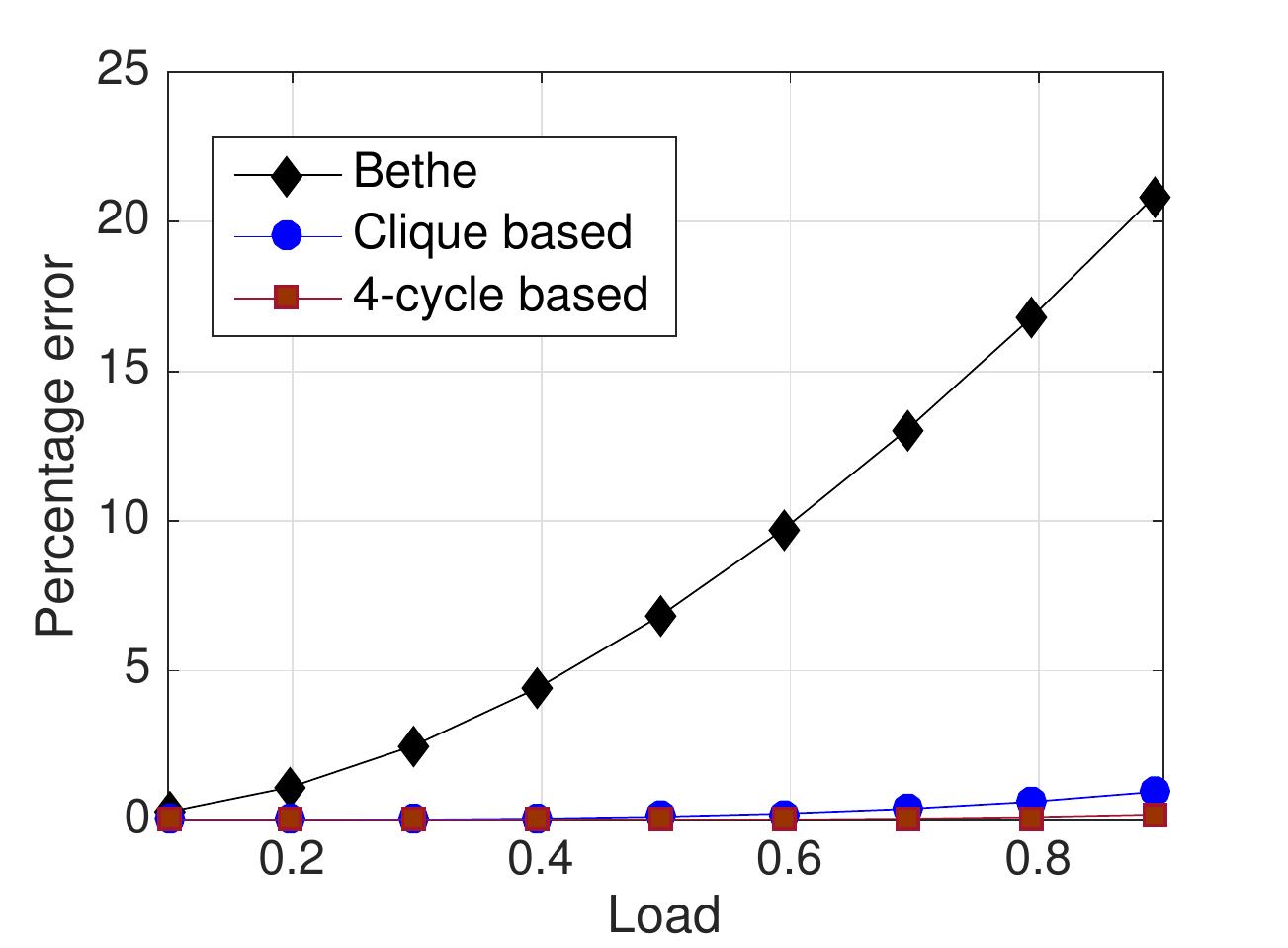} 
\caption{Error for 25-link topology}
\label{fig_r25_abs}
\end{minipage}
\caption*{The approximation error is plotted as function of the network load for 3 random topologies }
\end{figure*}

\emph{Simulation Setting:}
We first provide simulation results for complete graph, grid graph, and chordal graph (Figure \ref{fig_complete}- \ref{fig_chordal}) topologies. Further, we consider some random topologies of different sizes (Figure \ref{fig_random15} - \ref{fig_random25}). Specifically, we generate random geometric graphs on a two-dimensional square of length three. Two vertices are connected by an edge if they are within a distance of $0.8$. We consider symmetric service rate requirements for all the links\footnote{There is no specific reason to choose these parameters. The performance of the algorithms is qualitatively similar for other choices of parameters too.}.
 
 \subsection*{Approximation error}
We define the approximation error $e(s^t)$, as the maximum deviation from the required service rate. In particular, for a given target service rate vector $s^t=[s_i^t]_{i=1}^N$, $e(s^t)=\max\limits_{i}  \;|s^t_i - s^a_i| , $ where, $s^a=[s_i^a]_{i=1}^N$ are the service rates supported by using the approximated fugacities $\{\vt\}_{i=1}^N$. We vary the load (the ratio of the required service rate to the maximum permissible service rate) of the network by increasing the required service rates, and plot the percentage error. We compared the accuracy of our algorithms with the existing Bethe approximation based algorithm \cite{bethe_jshin}. 

\emph{Chordal graph:}
As proved in Theorem \ref{thm_chordal}, our algorithm is exact if the underlying conflict graph is chordal. This result is verified using a complete graph topology (which is trivially a chordal graph) and another chordal graph shown in Figure \ref{fig_chordal}. As shown in Figures \ref{fig_complete_prac}, \ref{fig_chordal_prac}, the Bethe approximation based algorithm incurs an error of about $20$ percent when the network is operated at the maximum capacity, while our Clique based algorithm is exact. 

\emph{Grid graph:} For the grid graph, the error is plotted in Figure \ref{fig_grid_prac}. Since, there are no cliques other than the trivial edges and vertices, the Bethe approximation and the Clique based algorithm will coincide. Both these algorithms incur an error of about $22$ percent at a load of $0.7$. However, our  4-cycle based algorithm is accurate with an error of about $1$ to $2$ percent.

\emph{Random graphs:} We considered three random graphs of sizes 15, 20 and 25 as shown in Figure \ref{fig_random15} - \ref{fig_random25}. We plotted the corresponding approximation errors in Figure \ref{fig_r15_abs} - \ref{fig_r25_abs}. It can be observed that the proposed algorithms perform significantly better than the Bethe approximation. In particular, for the twenty link topology at a load of $0.8$, the Bethe approximation has an error of about 28 percent, while our clique based and 4-cycle based algorithms have an error of about 9 and 7 percent respectively. Further, for the 15, 25 link topologies, when the Bethe approximation has an error of about 20 percent, the clique based approach incurs less than 2 percent error. Moreover, the 4-cycle based approach is almost exact with error close to zero. 

%
To present the average behaviour of the error, we considered thirty randomly generated topologies of size $20$, and observed the error when the required service rates are set to a load of $0.8$. We have presented the average error obtained across these topologies. As observed from Table \ref{table_error}, our algorithms have a significantly better error performance.

\begin{table}
\begin{center}

  \begin{tabular}{|c | c | }
  \hline
  Algorithm &  Average  \\  
     & error $\%$ \\ 
    \hline
Bethe  &  25.63 \\   
Clique-based & 2.78  \\
4-cycle based &1.83  \\
    \hline
     \end{tabular}
 \caption{The results obtained by simulating random topologies of size 20 at a load of $0.8$.}
  \label{table_error}
  \end{center}
 
   \end{table}


\noindent \emph{Remark on complexity}: The Bethe approximation \cite{bethe_jshin} uses only the service rates of the neighbours to compute the fugacities, and hence the error is significantly large. Our algorithms result in better error performance at the additional cost of obtaining the local topology information. Specifically, our clique-based algorithm requires the topology of the one-hop neighbourhood to compute the clique regions. The 4-cycle based algorithm requires the information of two-hop neighbourhood topology to compute the 4-cycle regions. Since the complexity of our algorithms just depend on the size of the local neighbourhood, the complexity does not scale with the network size, when the average density of the network is constant. On the other hand, stochastic gradient descent based algorithms \cite{libin} that converge to the exact fugacities, incur an exponentially slow convergence rate in the size of the network. 

\section{Conclusions}
The problem of computing the optimal fugacities for Gibbs sampling based CSMA algorithms is NP-hard. In this work, we derived estimates of the fugacities by using a framework called the \emph{regional free energy approximations}. Specifically, we derived explicit expressions for approximate fugacities corresponding to a given feasible service rate vector. We further proved that our approximate fugacities are exact for the class of \emph{chordal} graphs.  A distinguishing feature of our work is that the regional approximations that we proposed are tailored to conflict graphs with small cycles, which is a typical characteristic of wireless networks. Numerical results indicate that the fugacities obtained by the proposed methods are quite accurate, and significantly outperform the existing Bethe approximation based techniques.

%

\section{Proofs} \label{sec_proofs}
\subsection{Proof of Theorem \ref{thm_step1}} \label{proof_thm_step1}
%

The proof outline is as follows:
\begin{enumerate}
\item We define a Lagrangian function and characterize the stationary points of the RFE.
\item We derive the Lagrangian for the optimization problem \eqref{eq_opt_step1}.
\item We compare these two Lagrangian functions to prove the property (P2).
\end{enumerate}
\emph{Remark:} In the proof, $r \subset q$ denotes that $r$ is a \emph{strict} subset of $q$. The same is applicable for superset.

Let $\{b_r(\x_r)\}$ be a set of locally consistent regional distributions. Then the average energy \eqref{eq_averageenergy} can be computed as
\begin{align*}
U_{\Re}(\{b_r\}; \v)&= - \sum_{r \in {\Re}} c_r \E_{b_r}\Big[\sum_{j \in r} v_j x_j\Big], \\
&=-\sum_{r \in {\Re}} c_r \sum_{j \in r} v_j \E_{b_r}[x_j], \\
&=-\sum_{r \in {\Re}} c_r \sum_{j \in r} v_j \Big(\sum_{\x_r : x_j=1} b_r(\x_r)\Big), \\
&\overset{(a)}{=} - \sum_{r \in {\Re}} c_r \sum_{j \in r} v_j b_j(1)  , \\
&= -\sum_{j=1}^N \Big(\sum_{\{r \in \Re \;|\; j \in r\}} c_r \Big) v_j b_j(1),\\
&\overset{(b)}{=} -\sum_{j=1}^N v_j b_j(1) \numberthis \label{eq_u}.
\end{align*}
Here the equality (a) follows from the local consistency condition of $\{b_r(\x_r)\}$, (b) follows from the property of the counting numbers \eqref{eq_counting_general}. Substituting \eqref{eq_u} in the definition of the RFE \eqref{eq_RFE}, we obtain
\begin{align}
F_{\Re} \left(\{b_r\}; \v \right) &= -\sum_{j=1}^N v_j b_j(1) + \sum_{r \in \Re} c_r  \sum_{\x_r \in \I_r}b_r(\x_r) \ln b_r(\x_r). \label{eq_rfe_new}
\end{align}

Now, to derive the conditions for stationary points of RFE, we need to define a Lagrangian $\L$ for the RFE \eqref{eq_rfe_new}, which enforces the local consistency constraints. To that end, for every $r,q \in \Re$ such that $r \subset q$, we use the Lagrange multipliers $\l_{qr}(\x_r)$ to enforce the constraints
\begin{align*}
\sum_{\x_q \setminus \x_r} b_q(\x_q) = b_r(\x_r), \; \x_r \in \I_r, \text{ and }r, q \in \Re \text{ such that } r \subset q.
\end{align*}
Further, to enforce the constraint $\sum_{\x_r \in \I_r} b_r(\x_r) =1$ for the regional distribution at a region $r$, we use the Lagrange multiplier $\gamma_r$. The resulting Lagrangian is
\begin{align}
\begin{split}
&\L\left(\{b_r\}, \{\l_{qr}(\x_r)\}, \{\gamma_r\}\right) \\
&=F_{\Re}(\{b_r\}; \v) +\sum_{r \in \Re} \gamma_r \left( \sum_{\x_r \in \I_r} b_r(\x_r) -1\right)\\
&\; \; +\sum_{\{r,q \in \Re | r \subset q\}} \sum_{\x_r \in \I_r} \l_{qr}(\x_r) \left(\sum_{\x_q \setminus \x_r} b_q(\x_q) - b_r(\x_r) \right).  
\end{split} \label{eq_lag}
\end{align}
By setting the partial derivative of $\L\left(\{b_r\}, \{\l_{qr}(\x_r)\}, \{\gamma_r\}\right)$ with respect to $\{b_r(\x_r)\}$ to zero, we obtain the conditions for stationary points \eqref{eq_i1}-\eqref{eq_r_general}. 
Specifically, setting the partial derivatives of $\L\left(\{b_r\}, \{\l_{qr}(\x_r)\}, \{\gamma_r\}\right)$ with respect to the singleton distributions $b_i(1)$, and $b_i(0)$ to zero, gives
\begin{align}
-v_i + c_i \left(1+ \ln b_i(1) \right) + \gamma_i - \sum_{\{q \in \Re | \{i\} \subset q\}} \l_{qi}(1)=0, \; i \in \N, \label{eq_i1} \\ 
c_i \left(1+ \ln b_i(0) \right) + \gamma_i - \sum_{\{q \in \Re | \{i\} \subset q\}} \l_{qi}(0)=0, \; i \in \N. \label{eq_i0}
\end{align}
Similarly, for a region $r\in \Re$ which is not a singleton region, setting the partial derivative of $\L\left(\{b_r\}, \{\l_{qr}(\x_r)\}, \{\gamma_r\}\right)$ with respect to $b_r(\x_r)$ to zero, we obtain
\begin{align}
\begin{split}
&c_r \left(1+ \ln b_r(\x_r) \right) + \gamma_r  -\sum_{\{q \in \Re | r \subset q\}} \l_{qr}(\x_r) \\
&+ \sum_{\{p \in \Re | p \subset r\}} \l_{rp}(\x_p) =0, \; \;\;\forall \x_r \in \I_r, r \in \Re^{'},
\end{split} \label{eq_r_general}
\end{align}
where $\Re^{'}$ denote the set of all the regions $r \in \Re$ except the singleton regions.

Now, we consider the optimization problem \eqref{eq_opt_step1}, and define the corresponding Lagrangian.  Firstly, in the optimization problem \eqref{eq_opt_step1}, observe that the constraints \eqref{eq_bi_si} related to the singleton distributions $\{b_i\}$  can be absorbed into the local consistency constraints \eqref{eq_consistency}. In particular, the constraints 
\begin{align*}
\sum_{x_q \setminus \{x_i\}} b_q (x_q) &= b_i(x_i), \; \; x_i \in \{0,1\}, \text{ and } q \in \Re \text{ s.t. } \{i\} \subset q,\\
b_i(1)&=s_i; \;b_i(1)+b_i(0)=1 ,
\end{align*}
can be absorbed into the local consistency constraints \eqref{eq_absorb3}-\eqref{eq_absorb4} as follows:
\begin{align}
\sum_{x_q : x_i=1} b_q (x_q) &= s_i, \; \; q \in \Re \text{ s.t. } \{i\} \subset q, \label{eq_absorb3}\\
\sum_{x_q : x_i=0} b_q (x_q) &= 1-s_i, \; \; q \in \Re \text{ s.t. } \{i\} \subset q. \label{eq_absorb4}
\end{align}
After this modification of the constraints, the optimization problem \eqref{eq_opt_step1} will not contain the variables corresponding to $\{b_i\}$. In other words, it contains only $\{b_r\}_{r \in \Re^{'}}$, where $\Re^{'}$ denotes the set of regions which includes all the regions $r \in \Re$ except for the singleton regions. Note that when we simply write $\{b_r\}$, it refers to $\{b_r\}_{r \in \Re}$, \ie, the set of all the regional distributions including the singleton distributions.

Now, we define the Lagrangian $\G$ for the optimization problem \eqref{eq_opt_step1}. While enforcing the constraints, we replace one of its constraints, namely the constraint \eqref{eq_bi_si} with the equivalent modified constraints \eqref{eq_absorb3}-\eqref{eq_absorb4}.
\begin{align}
\begin{split}
&\G\left(\{b_r\}_{r \in \Re^{'}}, \{\l_{qr}(\x_r)\}, \{\gamma_r\}_{r \in \Re^{'}}\right) \\
&=\sum_{r \in \Re^{'}} c_r  \sum_{\x_r \in \I_r}b_r(\x_r) \ln b_r(\x_r) +\sum_{r \in \Re^{'}} \gamma_r \Big( \sum_{\x_r \in \I_r} b_r(\x_r) -1\Big)\\
&\; \; +\sum_{\{r,q \in \Re^{'} | r \subset q\}} \sum_{\x_r \in \I_r} \l_{qr}(\x_r) \left(\sum_{\x_q \setminus \x_r} b_q(\x_q) - b_r(\x_r) \right) \\
&\; \;+\sum_{q \in \Re^{'}} \sum_{\{i \;| \{i\} \subset q\}} \l_{qi}(1) \left(\sum_{\x_q : x_i=1} b_q(\x_q) - s_i \right) \\
&\; \;+\sum_{q \in \Re^{'}} \sum_{\{i \;| \{i\} \subset q\}}  \l_{qi}(0) \left(\sum_{\x_q : x_i=0} b_q(\x_q) - (1-s_i) \right). 
\end{split} \label{eq_lag_h}
\end{align}
Here,  the Lagrange multipliers $\{\l_{qr}(\x_r)\}$ enforce the local consistency constraints \eqref{eq_consistency}, and $\{\gamma_r\}_{r \in \Re^{'}}$ enforce the normalization constraints. For some $r \in \Re^{'}$, setting the partial derivative of $\G\left(\{b_r\}_{r \in \Re^{'}}, \{\l_{qr}(\x_r)\}, \{\gamma_r\}_{r \in \Re^{'}}\right) $ with respect to $b_r(\x_r)$ to zero, we obtain

\begin{align}
\begin{split}
&c_r \left(1+ \ln b_r(\x_r) \right) + \gamma_r  -\sum_{\{q \in \Re^{'} | r \subset q\}} \l_{qr}(\x_r) \\
&+ \sum_{\{p \in \Re^{'} | p \subset r\}} \l_{rp}(\x_p) + \sum_{\{i \;| \{i\} \subset r\}} \l_{ri}(x_i) =0, \;  r \in \Re^{'}.
\end{split} \label{eq_r_h}
\end{align}
Observe that the sum of the last two summations in \eqref{eq_r_h} is equal to $ \sum_{\{p \in \Re | p \subset r\}} \l_{rp}(\x_p)$. Hence, \eqref{eq_r_h} is same as \eqref{eq_r_general}. In other words, we have  
\begin{align}
\begin{split}
&\frac{\partial \L\left(\{b_r\}, \{\l_{qr}(\x_r)\}, \{\gamma_r\}\right) }{\partial b_r(\x_r)} \\
&=\frac{ \partial \G\left(\{b_r\}_{r \in \Re^{'}}, \{\l_{qr}(\x_r)\}, \{\gamma_r\}_{r \in \Re^{'}}\right)}{\partial b_r(\x_r)}, \forall \x_r \in \I_r, r \in \Re^{'},\label{eq_lag_eq}
\end{split}
\end{align}
where $\L\left(\{b_r\}, \{\l_{qr}(\x_r)\}, \{\gamma_r\}\right)$ \eqref{eq_lag} is the Lagrangian corresponding to the RFE. 

Now, let $\{b_r^*\}_{r \in \Re^{'}}$ correspond to a local optimal point of the optimization problem \eqref{eq_opt_step1}. Further define
\begin{align*}
b_i^*(1)&=s_i, \; \; \forall i \in \N,\\
b_i^*(0)&=1-s_i,\; \; \forall i \in \N,
\end{align*}
and extend $\{b_r^*\}_{r \in \Re^{'}}$ to $\{b_r^*\}$. Now, we want to show that $\{b_r^*\}$ is a stationary point of the RFE \eqref{eq_rfe_new} for some fugacities. For that, we have to find some fugacities $\{\tilde{v}_i\}$, and Lagrange multipliers $\left(\{\l_{qr}^*(\x_r)\}, \{\gamma^*_r\}\right)$ such that  $\left(\{b_r^*\}, \{\l_{qr}^*(\x_r)\}, \{\gamma_r^*\}\right)$ satisfies the stationarity conditions \eqref{eq_i1} - \eqref{eq_r_general} of the RFE. 

Using the fact that $\{b_r^*\}_{r \in \Re^{'}}$ is a stationary point (local optimizer) of the optimization problem \eqref{eq_opt_step1}, we are guaranteed to have Lagrange multipliers $\left(\{\l_{qr}^*(\x_r)\}, \{\gamma^*_r\}_{r \in \Re^{'}}\right)$ such that $\left(\{b_r^*\}_{r \in \Re^{'}}, \{\l_{qr}^*(\x_r)\}, \{\gamma_r^*\}_{r \in \Re^{'}}\right)$ satisfies the stationarity conditions \eqref{eq_r_h} of the Lagrangian  $\G\left(\{b_r\}_{r \in \Re^{'}}, \{\l_{qr}(\x_r)\}, \{\gamma_r\}_{r \in \Re^{'}}\right)$. Further due to the equality of the stationary conditions shown in \eqref{eq_lag_eq}, we conclude that $\left(\{b_r^*\}_{r \in \Re^{'}}, \{\l_{qr}^*(\x_r)\}, \{\gamma_r^*\}_{r \in \Re^{'}}\right)$  satisfies \eqref{eq_r_general}.

Now, inspired by \eqref{eq_i1}, \eqref{eq_i0}, we define
\begin{align}
\gamma_i^* &= \sum_{q \in \Re | \{i\} \subset q} \l_{qi}^*(0)-c_i \left(1+ \ln b_i^*(0) \right) ,\; i \in \N, \label{eq_def_gamma}\\
\tilde{v_i} &= c_i \left(1+ \ln b_i^*(1) \right) + \gamma_i^* - \sum_{q \in \Re | \{i\} \subset q} \l^*_{qi}(1), \; i \in \N. \label{eq_def_v}
\end{align}
Hence, in the light of \eqref{eq_lag_eq}, and the definition of $\gamma_i^*, \tilde{v_i}$ \eqref{eq_def_gamma} - \eqref{eq_def_v}, it follows that $\left(\{b_r^*\}, \{\l_{qr}^*(\x_r)\}, \{\gamma_r^*\}\right)$ satisfies all the conditions \eqref{eq_i1} - \eqref{eq_r_general} required for stationarity of $\L\left(\{b_r\}, \{\l_{qr}(\x_r)\}, \{\gamma_r\}\right)$ \eqref{eq_lag}. Hence, $\{b_r^*\}$ is a stationary point of the RFE $F_{\Re}(\{b_r\}; \tilde{\v})$ \eqref{eq_rfe_new} for the fugacities $\{\tilde{v_i}\}$ defined in \eqref{eq_def_v}. Therefore, $\{b_r^*\}$ satisfies the property $(P2)$ defined in Section \ref{sec_inverse}.

\subsection{Proof of Theorem \ref{thm_step2}} \label{proof_thm_step2}
Consider the stationarity conditions \eqref{eq_i1} - \eqref{eq_r_general} of the Lagrangian $\L\left(\{b_r\}, \{\l_{qr}(\x_r)\}, \{\gamma_r\}\right)$ \eqref{eq_lag}.
Subtracting \eqref{eq_i0} from \eqref{eq_i1} gives
\begin{align}
c_i \ln \left( \frac{b_i(1)}{b_i(0)}\right)  = v_i - \sum_{q \in \Re | \{i\} \subset q} \left(\l_{qi}(0) - \l_{qi} (1)\right). \label{eq_i}
\end{align}
Let $r \supset \{i\}$ be some non-singleton region containing $i$. Let $\x_r^i$ denote the argument $\x_r$ such that $x_i=1$, $x_j=0, \forall j \in r \setminus \{i\}$. Similarly for any $r \in \Re$, let ${\bf{0}}$ denote $\x_r$ such that $x_j=0, \forall j \in r$. Then by setting the derivative of $\L\left(\{b_r\}, \{\l_{qr}(\x_r)\}, \{\gamma_r\}\right)$ given in \eqref{eq_lag} with respect to $b_r(\x_r^i)$ to zero, we obtain
\begin{align}
\begin{split}
&c_r \left(1+ \ln b_r(\x_r^i) \right) + \gamma_r  -\sum_{\{q \in \Re | r \subset q\}} \l_{qr}(\x_r^i) \\
&+ \sum_{\{p \in \Re | p \subset r, i \in p\}} \l_{rp}(\x_p^i) + \sum_{\{p \in \Re | p \subset r, i \notin p\}} \l_{rp}({\bf{0}}) =0.
\end{split} \label{eq_r1}
\end{align}
Similarly, setting the derivative of $\L\left(\{b_r\}, \{\l_{qr}(\x_r)\}, \{\gamma_r\}\right)$ given in \eqref{eq_lag} with respect to $b_r({\bf{0}})$ to zero, we obtain
\begin{align}
\begin{split}
&c_r \left(1+ \ln b_r({\bf{0}}) \right) + \gamma_r  -\sum_{\{q \in \Re | r \subset q\}} \l_{qr}({\bf{0}}) \\
&+ \sum_{\{p \in \Re | p \subset r, i \in p\}} \l_{rp}({\bf{0}}) + \sum_{\{p \in \Re | p \subset r, i \notin p\}} \l_{rp}({\bf{0}}) =0.
\end{split} \label{eq_r0}
\end{align}
Now, subtracting \eqref{eq_r0} from \eqref{eq_r1} we obtain 
\begin{align}
\begin{split}
&c_r \ln \left( \frac{b_r(\x_r^i)}{b_r({\bf{0}})}\right) =- \sum_{\{q \in \Re | r \subset q\}} \left(\l_{qr}({\bf{0}}) - \l_{qr} (\x_r^i)\right) \\
&\; \; - \sum_{\{p \in \Re | p \subset r, i \in p\}} \left( \l_{rp}(\x_p^i) - \l_{rp}({\bf{0}})\right).
\end{split} \label{eq_r}
\end{align}
For shorthand notation, for a given $i \in V$, let us define
\begin{align*}
\beta_{qr}:=\left(\l_{qr}({\bf{0}}) - \l_{qr} (\x_r^i)\right), \; \; \forall q,r \ni i.
\end{align*}
Then \eqref{eq_r} can be written as 
\begin{align}
&c_r \ln \left( \frac{b_r(\x_r^i)}{b_r({\bf{0}})}\right) \\
&= \sum_{\{q \in \Re | i \in q\}} \Big({\i(r \supset q)} \beta_{rq} - {\i(r \subset q)} \beta_{qr} \Big), \; \forall r \supset \{i\}. \label{eq_r_new}
\end{align}

Next, let us consider the expression $\sum_{\{r \in \Re | i \in r\}}  c_r \ln \left( \frac{b_r(\x_r^i)}{b_r({\bf{0}})}\right)$ and split it as
\begin{align*}
&\sum_{\{r \in \Re | i \in r\}}  c_r \ln \left( \frac{b_r(\x_r^i)}{b_r({\bf{0}})}\right) \\
&= c_i \ln \left( \frac{b_i(1)}{b_i(0)}\right) + \sum_{\{r \in \Re | r \supset \{i\}\}} c_r \ln \left( \frac{b_r(\x_r^i)}{b_r({\bf{0}})}\right). \numberthis \label{eq_wewant}
\end{align*}
If we substitute \eqref{eq_i} and \eqref{eq_r_new} in \eqref{eq_wewant}, we obtain 
\begin{align}
&\sum_{\{r \in \Re | i \in r\}}  c_r \ln \left( \frac{b_r(\x_r^i)}{b_r({\bf{0}})}\right) \nonumber\\
&= v_i - \sum_{\{q \in \Re | \{i\} \subset q\}} \beta_{qi} \nonumber\\&
\; \;+ \sum_{\{r \in \Re | r\supset \{i\}\}} \sum_{\{q \in \Re | i \in q\}} \Big({\i(r \supset q)} \beta_{rq} - {\i(r \subset q)} \beta_{qr} \Big), \nonumber \\
&\overset{(a)}{=} v_i + \sum_{\{q \in \Re | i \in q\}} \Big({\i(\{i\} \supset q)} \beta_{iq} - {\i(\{i\} \subset q)} \beta_{qi} \Big) \nonumber\\&
\; \;+ \sum_{\{r \in \Re | r\supset \{i\}\}} \sum_{\{q \in \Re | i \in q\}} \Big({\i(r \supset q)} \beta_{rq} - {\i(r \subset q)} \beta_{qr} \Big), \nonumber \\
&= v_i  + \sum_{\{r \in \Re | i \in r\}} \sum_{\{q \in \Re | i \in q\}} \Big({\i(r \supset q)} \beta_{rq} - {\i(r \subset q)} \beta_{qr} \Big), \nonumber\\
&\overset{(b)}{=}v_i. \label{eq_last}
\end{align}
Here, the equality $(a)$ is obtained by simply rewriting the term $-\sum_{\{q \in \Re | \{i\} \subset q\}} \beta_{qi}$ into the equivalent form $\sum_{\{q \in \Re | i \in q\}} \Big({\i(\{i\} \supset q)} \beta_{iq} - {\i(\{i\} \subset q)} \beta_{qi} \Big)$. Further, the equality $(b)$ is obtained by observing that all the terms other than $v_i$ cancel out, and evaluate to zero.

Hence, from \eqref{eq_last}, we can conclude that if $\{b_r^*\}$ is a stationary point of $F_{\Re}(\{b_r\}; \v)$, which is locally consistent, we have
\begin{align*}
\exp(v_i)= \prod_{\{r \in \Re | i \in r\}} \left( \frac{b_r^*(\x_r^i)}{b_r^*({\bf{0}})}\right)^{c_r}, \; \forall i \in \N. 
\end{align*}

\subsection{Proof of Theorem \ref{thm_chordal}} \label{proof_chordal}
\begin{figure}
\begin{center}
  \hspace{-2mm}
\begin{tikzpicture}[scale=0.7]

\fill (1,2)circle(0.06cm);
\fill (3,2)circle(0.06cm);
\fill (4,2)circle(0.06cm);
\fill (1,1)circle(0.06cm);
\fill (3,1)circle(0.06cm);
\fill (4,1)circle(0.06cm);

\node at (.9,2.2) {\begin{scriptsize}4\end{scriptsize}};
\node at (1.1,.8) {\begin{scriptsize}1\end{scriptsize}};
\node at (3,2.2) {\begin{scriptsize}5\end{scriptsize}};
\node at (2.9,0.8) {\begin{scriptsize}2\end{scriptsize}};
\node at (4.1,2.2) {\begin{scriptsize}6\end{scriptsize}};
\node at (4.1,0.8) {\begin{scriptsize}3\end{scriptsize}};

\draw (1,2)-- (3,2) -- (3,1) -- (1,1) -- (1,2);  

\draw (3,2)-- (4,2) -- (4,1) -- (3,1) -- (3,2); 


;

\draw (1,1) -- (3,2); 

\draw (4,1) -- (3,2); 

\end{tikzpicture}
  \vspace{-2mm}
 \caption{An example of a chordal graph} 
  \label{fig_chordal_example}
  \end{center}
  \end{figure}
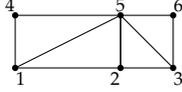
  
  \begin{figure}
  \centering
  \begin{tikzpicture}[scale=1]

\node at (2.7,2) {Valid Junction tree:};
\node at (2.7,1) {Invalid Junction tree:};

\node (rect145) at (5,2) [draw,minimum width=.25cm,minimum height=.25cm] {\begin{scriptsize}1,4,5\end{scriptsize}};
\node (rect125) at (6,2) [draw,minimum width=.25cm,minimum height=.25cm] {\begin{scriptsize}1,2,5\end{scriptsize}};
\node (rect235) at (7,2) [draw,minimum width=.25cm,minimum height=.25cm] {\begin{scriptsize}2,3,5\end{scriptsize}};
\node (rect356) at (8,2) [draw,minimum width=.25cm,minimum height=.25cm] {\begin{scriptsize}3,5,6\end{scriptsize}};

  \path (rect145) edge (rect125);
    \path (rect125) edge (rect235);
        \path (rect235) edge (rect356);
 
 \node (rect1450) at (5,1) [draw,minimum width=.25cm,minimum height=.25cm] {\begin{scriptsize}1,2,5\end{scriptsize}};
\node (rect1250) at (6,1) [draw,minimum width=.25cm,minimum height=.25cm] {\begin{scriptsize}1,4,5\end{scriptsize}};
\node (rect2350) at (7,1) [draw,minimum width=.25cm,minimum height=.25cm] {\begin{scriptsize}2,3,5\end{scriptsize}};
\node (rect3560) at (8,1) [draw,minimum width=.25cm,minimum height=.25cm] {\begin{scriptsize}3,5,6\end{scriptsize}};

  \path (rect1450) edge (rect1250);
    \path (rect1250) edge (rect2350);
        \path (rect2350) edge (rect3560);

\end{tikzpicture}
\captionof{figure}{Illustration of Junction tree representation for Figure \ref{fig_chordal_example}, with maximal cliques as regions. The top figure is a valid junction tree, while the bottom figure is not a valid junction tree.}
  \label{fig_juntree}
\end{figure}
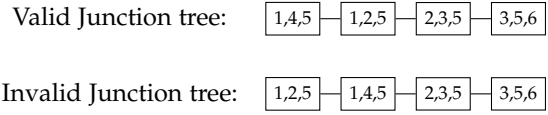
We now introduce the notion of a junction tree which is required for this proof.
\begin{definition} (Junction tree)
Let $\Re$ denote a given collection of regions, and $\Re_0$ denote the maximal regions of $\Re$. A junction tree $T = (\Re_0, \mathcal{E})$ is a tree in which the nodes correspond to the maximal regions, and the edges are such that they satisfy the following Running intersection property: For any two maximal regions $r, q \in \Re_0$, the elements in $r \cap q$ should be part of all the maximal regions on the unique path from $r$ to $q$ in the tree $T$.
\end{definition}
It is a known fact that, if we consider a chordal graph with maximal cliques as the collection of regions, then we can construct a junction tree \cite[Page 30]{book_martin}. The definition of Junction tree is illustrated with the example of a chordal graph in Figure \ref{fig_chordal_example}, \ref{fig_juntree}. In Figure \ref{fig_juntree}, the top tree is a valid junction tree, since every pair of regions satisfy the running intersection property. For example, consider the pair of regions $\{1,4,5\}$, $\{2,3,5\}$. The region $\{1,2,5\}$ which is in the path joining them, contains their intersection $\{1,5\}$. Similarly, it is easy to observe that the figure in the bottom is not a valid junction tree. For example, consider the pair of regions $\{1,2,5\}$, $\{2,3,5\}.$ The region $\{1,4,5\}$ which is in the path joining them, does not contain their intersection $\{2,5\}$, and hence violates the running intersection property.


Next, we require the following result from \cite[page 101]{book_martin}:  ``Consider a collection of regions $\Re$ that is closed under intersection. Assume that there exists a junction tree for the collection $\Re$. Then, it is known that a product form distribution $p(\x)$ of the form given in \eqref{eq_dist}, can be factorized\footnote{This factorization is presented in \cite[page 101]{book_martin}, and is referred to as Hypertree based Re-parametrization. The result is presented using the terminology of mobius function. It can be easily argued that the factorization that we use here is equivalent to that in \cite{book_martin}.} in terms of its exact regional distributions as
$p(\x)= \prod_{r \in \Re}p_r(x_r)^{c_r}.$"


As discussed earlier, there exists a junction tree representation for the maximal cliques of a chordal graph \cite[Page 30]{book_martin}. Hence, using the above result, we can factorize $p(\x)$ defined in \eqref{eq_dist} in terms of the collection of cliques $\Re$ as follows:
\begin{align}
p(\x)=\frac{1}{Z(v)} \exp\Big(\sum_j v_j x_j\Big)= \prod_{r \in \Re}p_r(x_r)^{c_r}, \; \forall \x \in \I. \label{eq_hr1}
\end{align}
Substituting $\x=\bf{0}$ (\ie, $x_j=0$ for all $j \in \N$) in \eqref{eq_hr1} gives
\begin{align*}
Z(v)^{-1} = \prod_{r \in \Re} p_r({\bf{0}})^{c_r}.
\end{align*}
Now using the observations made in \eqref{eq_reg_dist}, we obtain
\begin{align}
Z(v)^{-1} = \prod_{r \in \Re} \Big(1-\sum_{j \in r} p_j(1)\Big)^{c_r}. \label{eq_z}
\end{align}

Similarly, substituting $\x=[x_i]_{i=1}^N$ with $x_i=1$, $x_j=0$ for all $j \in \N \setminus \{i\}$ in \eqref{eq_hr1}, and using \eqref{eq_reg_dist}, we obtain
\begin{align}
\frac{\exp(v_i)}{Z(v)}  = \prod_{q \in \Re^i} (p_i(1))^{c_q} \prod_{r \in \Re \setminus \Re^i} \Big(1-\sum_{j \in r} p_j(1)\Big)^{c_r}, \label{eq_v}
\end{align}
where $\Re^i=\{r \in \Re \;|\; i \in r\}$ is the set of regions in which vertex $i$ is involved.

Now, dividing \eqref{eq_v} with \eqref{eq_z}, we obtain
\begin{align*}
\exp(v_i)=  \frac{\prod_{q \in \Re^i} (p_i(1))^{c_q}}{\prod_{r \in \Re^i} \Big(1-\sum_{j \in r} p_j(1)\Big)^{c_r}}.
\end{align*}

Further, from \eqref{eq_counting_general}, we know  that $\sum_{q\in \Re^i} c_q=1$. Hence,
\begin{align*}
\exp(v_i)= p_i(1) \prod_{r \in \Re^i} \Big(1-\sum_{j \in r} p_j(1)\Big)^{-c_r}. 
\end{align*}
Since the service rate $s_i$ is to be obtained as the marginal $p_i(1)$, the exact fugacities and the corresponding service rates are related by the above equation. It can observed that our estimate of fugacities \eqref{eq_main} follows the same relation, and hence our algorithm is exact for chordal graphs.

\bibliographystyle{IEEEtran}
\bibliography{myreferences_kikuchi}

\begin{IEEEbiography}[{\includegraphics[width=1in,height=1in]{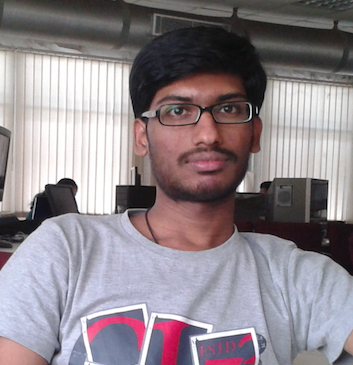}}]{Peruru Subrahmanya Swamy} obtained his B. Tech. degree in Electronics and Communication Engineering from Sastra University, Thanjavur, India in 2011. He is currently a graduate student in the Department of Electrical Engineering, IIT Madras, Chennai, India. He worked as a Project Associate in Analog Devices DSP Learning Centre, IIT Madras during 2011-2012. His research interests lie in communication networks, optimization, and stochastic geometry.
\end{IEEEbiography}

\begin{IEEEbiography}[{\includegraphics[width=1in,height=1in]{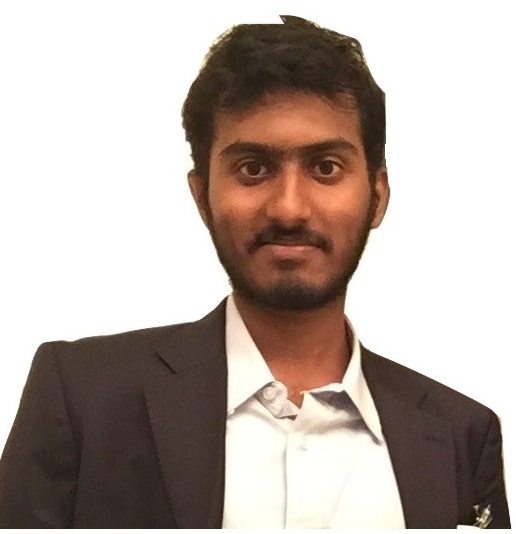}}]{Venkata Pavan Kumar Bellam} obtained his B. Tech. degree in Electrical Engineering from Indian Institute of Technology, Madras in 2016. He is currently pursuing his M.S. degree in Computer Engineering from Virginia Tech. He worked as a summer intern at Samsung Research and Development Institute, Bengaluru in 2015. His research interests lie in communication networks, optimization, machine learning and artificial intelligence.
\end{IEEEbiography}

\begin{IEEEbiography}[{\includegraphics[width=1in,height=1in]{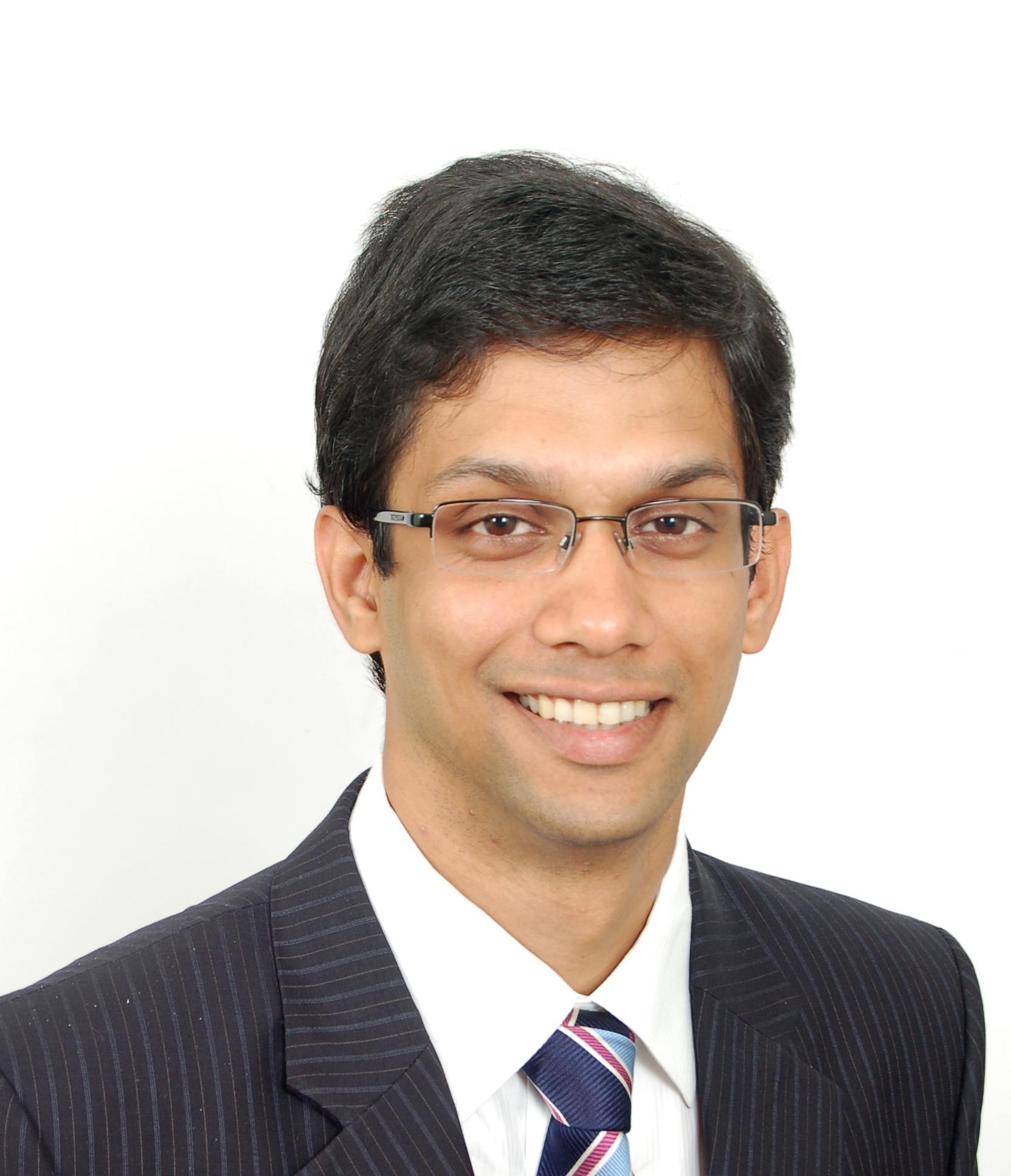}}]{Radha Krishna Ganti (S'0 - M'10)} is an Assistant Professor at the Indian Institute of Technology Madras, Chennai, India. He was a Postdoctoral researcher in the Wireless Networking and Communications Group at UT Austin from 2009-11. He received his B. Tech. and M. Tech. in EE from the Indian Institute of Technology, Madras, and a Masters in Applied Mathematics and a Ph.D. in EE from the University of Notre Dame in 2009. His doctoral work focused on the spatial analysis of interference networks using tools from stochastic geometry. He is a co-author of the monograph Interference in Large Wireless Networks (NOW Publishers, 2008). He received the 2014 IEEE Stephen O. Rice Prize,  and the 2014 IEEE Leonard G. Abraham Prize and the 2015  IEEE Communications society  young author best paper award.
\end{IEEEbiography}

\begin{IEEEbiography}[{\includegraphics[width=1in,height=1in]{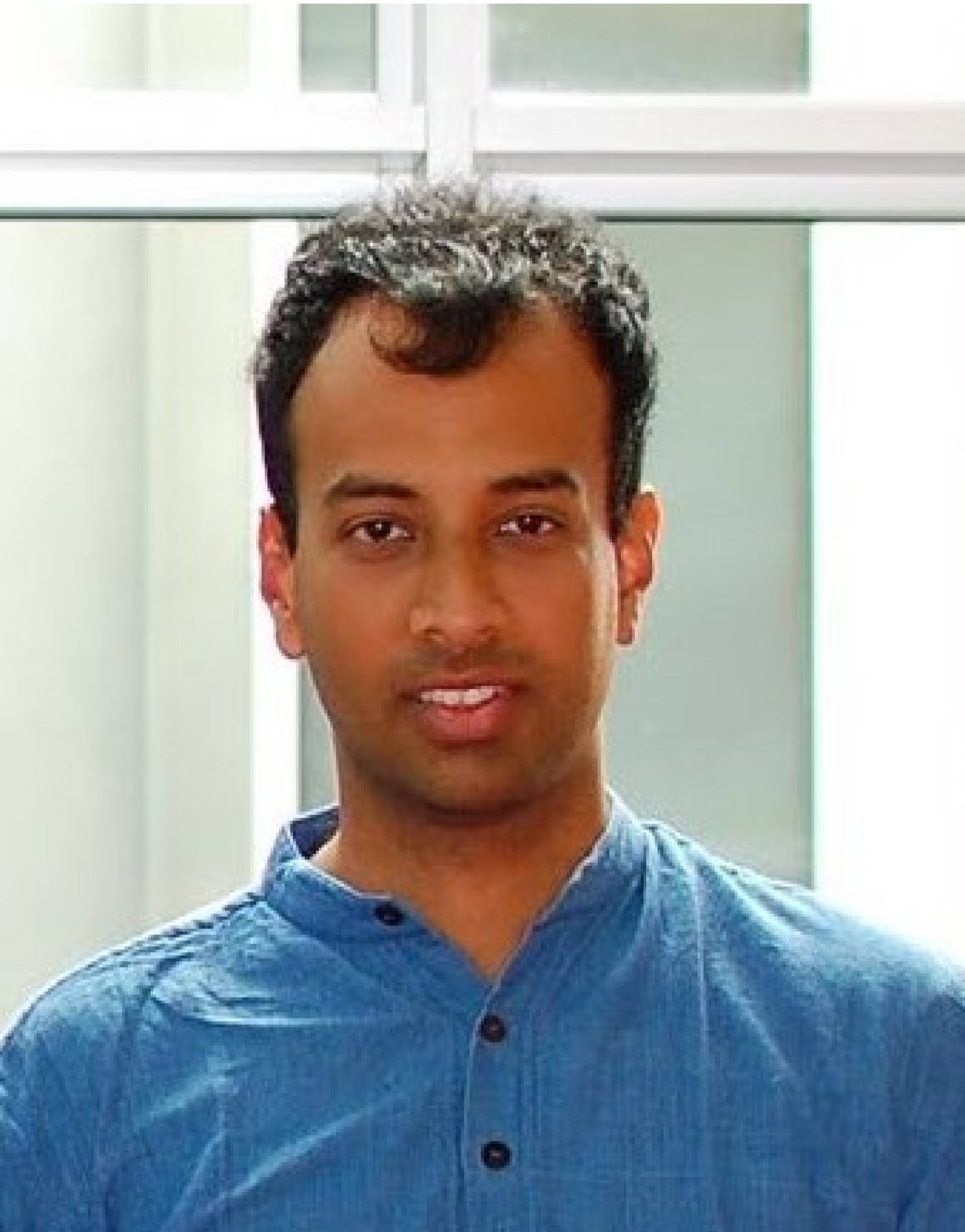}}]{Krishna Jagannathan} obtained his B. Tech. in Electrical Engineering from IIT Madras in 2004, and the S.M. and Ph.D. degrees in Electrical Engineering and Computer Science from Massachusetts Institute of Technology (MIT) in 2006 and 2010 respectively. During 2010-2011, he was a visiting post-doctoral scholar in Computing and Mathematical Sciences at Caltech, and an off-campus post-doctoral fellow at MIT. Since November 2011, he has been an assistant professor in the Department of Electrical Engineering, IIT Madras. He worked as a consultant at the Mathematical Sciences Research Center, Murray Hills, NJ in 2005, an engineering intern at Qualcomm, Campbell, CA in 2007. His research interests lie in the stochastic modeling and analysis of communication networks, transportation networks, network control, and queueing theory.
\end{IEEEbiography}


\end{document}